\documentclass[twocolumn,aps,preprintnumbers,showpacs,superscriptaddress]{revtex4-1}
\usepackage{latexsym}
\usepackage{amssymb}
\usepackage{amsmath}
\usepackage{amsthm}
\usepackage{amsfonts}
\usepackage{ifthen}
\usepackage{revsymb}
\usepackage{yfonts}
\usepackage{graphicx}
\usepackage{algpseudocode}

\usepackage[colorlinks = true]{hyperref}

\usepackage{xcolor}
\definecolor{darkred}  {rgb}{0.5,0,0}
\definecolor{darkblue} {rgb}{0,0,0.5}
\definecolor{darkgreen}{rgb}{0,0.5,0}

\hypersetup{
  urlcolor   = blue,         
  linkcolor  = darkblue,     
  citecolor  = darkgreen,    
  filecolor  = darkred       
}

\newcommand{\be}{\begin{equation}}
\newcommand{\ee}{\end{equation}}
\newcommand{\ba}{\begin{array}}
\newcommand{\ea}{\end{array}}
\newcommand{\bea}{\begin{eqnarray}}
\newcommand{\eea}{\end{eqnarray}}

\newcommand{\calH}{{\cal H }}
\newcommand{\calL}{{\cal L }}
\newcommand{\calM}{{\cal M }}
\newcommand{\calK}{{\cal K }}

\newcommand{\calP}{{\cal P }}

\newcommand{\calV}{{\cal V }}
\newcommand{\calW}{{\cal W }}
\newcommand{\calS}{{\cal S }}

\newcommand{\calG}{{\cal G }}

\newcommand{\calD}{{\cal D }}

\newcommand{\ZZ}{\mathbb{Z}}
\newcommand{\CC}{\mathbb{C}}

\newcommand{\EE}{\mathbb{E}}
\newcommand{\FF}{\mathbb{F}}

\newcommand{\prob}[1]{\mathrm{Pr}{\left[\, {#1} \, \right]}}

\newtheorem{dfn}{Definition}
\newtheorem{lemma}{Lemma}

\begin{document}

\title{Improved classical simulation of quantum circuits dominated by Clifford gates}

\author{Sergey \surname{Bravyi}}
\affiliation{IBM T.J. Watson Research Center,  Yorktown Heights  NY 10598}
\author{David \surname{Gosset}}
\affiliation{Walter Burke Institute for Theoretical Physics and Institute for Quantum Information and Matter, 
California Institute of Technology}

\date{\today}

\begin{abstract}
The Gottesman-Knill theorem asserts that a quantum
circuit composed of Clifford gates can be efficiently simulated on a classical computer. 
Here we revisit this theorem and extend it to quantum circuits
composed of Clifford and $T$ gates, where $T$ is the single-qubit $45^\circ$ phase shift.
We assume that the circuit outputs a bit string $x$ obtained by measuring some subset of $w$
qubits. Two simulation tasks are considered: 
(1)  computing the probability of a given output $x$, 
and (2) sampling $x$ from the output probability distribution.
It is shown that these tasks can be solved on a classical computer in time $poly(n,m)+2^{0.5 t} t^3$ and
$poly(n,m)+2^{0.23 t} t^3 w^3$ respectively, where
$t$ is the number of $T$-gates, $m$ is the total number of gates, and $n$ is the number of qubits.
The proposed simulation algorithms  may serve as a verification tool for medium-size quantum computations
that are dominated by Clifford gates.  
The main ingredient of both algorithms is a subroutine for approximating the norm of an $n$-qubit state which is given as a linear combination of $\chi$ stabilizer states. The subroutine runs in time  $\chi n^3 \epsilon^{-2}$, where $\epsilon$ is the relative error.  We  also  develop techniques for approximating tensor products of ``magic states"
by linear combinations of stabilizer states. To demonstrate the power of the new simulation methods, we performed a classical simulation of a hidden shift quantum algorithm 
with $40$ qubits, a few hundred Clifford gates, and nearly $50$ $T$-gates. 
\end{abstract}


\maketitle

\section{Introduction}
\label{sec:intro}

The path towards building a large-scale quantum computer will inevitably require verification and validation of small quantum devices. One way to check that such a device is working properly is to simulate it on a classical computer. This becomes impractical at some point because the cost of classical simulation typically grows
exponentially with the size of a quantum system. 
With this fundamental limitation in mind it is natural to ask how well we can do in practice.

Simulation methods which store a complete description of an $n$-qubit quantum state as a 
complex vector of size $2^n$  are limited to a small number of qubits $n\approx 30$. For example, a state-of-the art implementation has been used to simulate Shor's factoring algorithm with 31 qubits and roughly half a million gates \cite{liquid}. For certain restricted classes of quantum circuits it is possible to do much better \cite{gottesman99, aaronson2004improved, markov2008simulating, Nest2010, pashayan2015estimating}. Most significantly, the Gottesman-Knill theorem allows efficient classical simulation of quantum circuits composed of gates in the so-called Clifford group \cite{gottesman99}. In practice this allows one to simulate such circuits with thousands of qubits \cite{aaronson2004improved,liquid}. It also means that a quantum computer will need to use gates outside of the Clifford group in order to achieve useful speedups over classical computation. The full power of quantum computation can be recovered by adding a single non-Clifford gate to the Clifford group. A simple choice is the single-qubit $T=|0\rangle\langle 0| + e^{i\pi/4} |1\rangle\langle 1|$ gate; the Clifford+$T$ gate set obtained in this way is a natural instruction set for small-scale fault-tolerant quantum computers based on the surface code \cite{surfacecode,Fowler2012}, and has been at the centre of a recent renaissance in classical techniques for compiling quantum circuits \cite{KMM,Selinger, RossSelinger}.

In this paper we present two new  algorithms for classical  simulation of quantum circuits over the Clifford+$T$ gate set. The runtime of the algorithms is polynomial in the number of qubits and the number of Clifford gates in the circuit but exponential in the number of $T$ gates, or $T$-count. This exponential scaling is sufficiently mild that we anticipate a classical simulation of Clifford+$T$ circuits 
with a few hundred  qubits and $T$-count $t\le 50$ can be performed on
a medium-size computer cluster. Thus our algorithms may serve as a 
verification tool for small quantum computations dominated by Clifford gates. Such computations arise naturally if a
logical quantum circuit is realized fault-tolerantly
using some stabilizer code. The first demonstrations of logical quantum circuits using the surface code are likely to be dominated by Clifford gates due to a high implementation cost associated with logical $T$-gates~\cite{Fowler2013surface,Jones2013}. 

To describe our results let us fix some notation. A Clifford+$T$ quantum circuit of length $m$ acting on $n$ qubits is a
unitary operator $U=U_m \cdots U_2 U_1$, where each $U_j$ is a one- or two-qubit gate from the set $\{H,S,T,CNOT\}$  where $H$ is the  Hadamard gate and $S=|0\rangle\langle 0| + i|1\rangle\langle 1|$. We shall write $m=c+t$, where $c$ is the number of Clifford gates ($H,S,CNOT$)
and $t$ is the number of $T$-gates also known as the $T$-count.
Applying $U$ to the initial state $|0^n\rangle$
and measuring some fixed output register  $Q_{out} \subseteq [n]$ in the
$0,1$-basis generates a random bit string $x$ of length $w=|Q_{out}|$. A string $x$
appears with probability 
\begin{equation}
\label{Pout1}
P_{out}(x)=\langle 0^{n} | U^\dag \Pi(x) U |0^{ n}\rangle,
\end{equation}
where $\Pi(x)$  projects $Q_{out}$ onto the basis state $|x\rangle$
and acts trivially on the remaining qubits.

Our first result is a classical Monte Carlo algorithm that 
approximates the probability $P_{out}(x)$
for a given string $x\in \{0,1\}^w$ with a specified relative error $\epsilon$
and a failure probability $p_f$.
 The algorithm has runtime 
\begin{equation}
\label{T1}
\tau=O\left(  (w+t)(c+t) +(n+t)^3 + 2^{\beta t} t^3 \epsilon^{-2} \log(p_f^{-1})\right),
\end{equation}
where $\beta\le 1/2$ is a constant that depends on the implementation details. 
For example, assuming that $\epsilon$ and $p_f$ are some
fixed constants and $w\le t\le n\le c$, the runtime becomes
\[
\tau=O(n^3+ct + 2^{\beta t} t^3).
\]
Our second result is a classical algorithm that
allows one to sample the output string $x$ from a distribution which is
$\epsilon$-close to $P_{out}$ with respect to the $L_1$-norm. The sampling algorithm has runtime
\begin{equation}
\label{T2}
\tau=\tilde{O}\left( w(w+t)(c+t) + w(n+t)^3 + 2^{\gamma t} t^3 w^3 \epsilon^{-4} \right), 
\end{equation}
where the $\tilde{O}$ notation hides a factor logarithmic in $w$ and $\epsilon^{-1}$, and 
\begin{equation}
\label{T3}
\gamma\le -2\log_2{\left(\cos{(\pi/8)}\right)}\approx 0.228
\end{equation}
is a constant that
depends on the implementation details. 
We expect the sampling algorithm to be practical when $w$ is small and $\epsilon$ is not too small.
For example, assuming that 
the circuit outputs a single bit ($w=1$), $\epsilon$ is a fixed constant,
and $t\le n\le c$, the runtime becomes
\[
\tau=O(n^3 + ct + 2^{\gamma t} t^3).
\] 
Both algorithms can be divided into   independent subroutines with a runtime $O(t^3)$ each
and thus support a large amount of parallelism. 
We provide pseudocode for the main subroutines used in the algorithms and a timing analysis for
the MATLAB implementation~\footnote{{T}he MATLAB implementation
of the sampling algorithm is available upon request to the authors.}
in the Supplemental Material.

Since the simulation runtime is likely to be dominated by the 
terms exponential in $t$, one may wish to minimize 
the exponents $\beta$, $\gamma$ in Eqs.~(\ref{T1},\ref{T2}).
These exponents are related to the  stabilizer rank~\cite{BSS15} of a magic state
\[
|A\rangle = 2^{-1/2}( |0\rangle + e^{i\pi/4} |1\rangle).
\]
Recall that a $t$-qubit state is called a stabilizer state
if it has the form $V|0^t\rangle$, where $V$ is a  quantum circuit
composed of Clifford gates. 
Stabilizer states form an overcomplete basis in the Hilbert space
of $t$ qubits. Let $\chi_t(\delta)$ be the smallest integer $\chi$
such that $A^{\otimes t}$ can be approximated with an error
at most $\delta$ by a linear combination of $\chi$ stabilizer states (here the approximating state $\psi$ should satisfy $|\langle A^{\otimes t}|\psi\rangle|^2\geq 1-\delta$). The runtime scaling in Eq.~(\ref{T1}) holds for any exponent $\beta$
such that $\chi_t(0)=O(2^{\beta t})$ for all sufficiently large $t$.
Using the results of~\cite{BSS15} one can choose
$\beta = (1/6)\log_2{(7)} \approx 0.47$.
Stronger upper bounds on the stabilizer rank $\chi_t(0)$ could
improve the runtime scaling in Eq.~(\ref{T1}).
Likewise, the runtime scaling in 
Eq.~(\ref{T2}) holds for any exponent $\gamma$ such that
 $\chi_t(\delta)=O(2^{\gamma t})$ 
for any constant $\delta>0$ and all 
sufficiently large $t$.
For simplicity here we assumed that the precision parameter $\epsilon$
in Eq.~(\ref{T2}) is a constant. 
In this paper we propose a systematic method of finding approximate stabilizer decompositions
of $A^{\otimes t}$ which yields an upper bound 
$\chi_t(\delta)=O(2^{\gamma t} \delta^{-1})$, where $\gamma\approx 0.228$, see Eq.~(\ref{T3}).
We conjecture  that this upper bound is tight.

\begin{figure*}
\includegraphics[height=6cm]{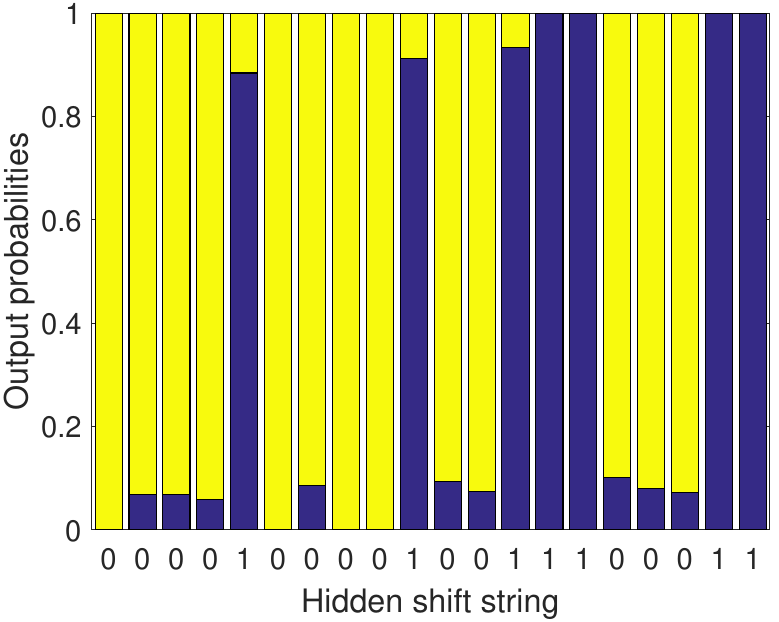}
$\quad\quad\quad$
\includegraphics[height=6cm]{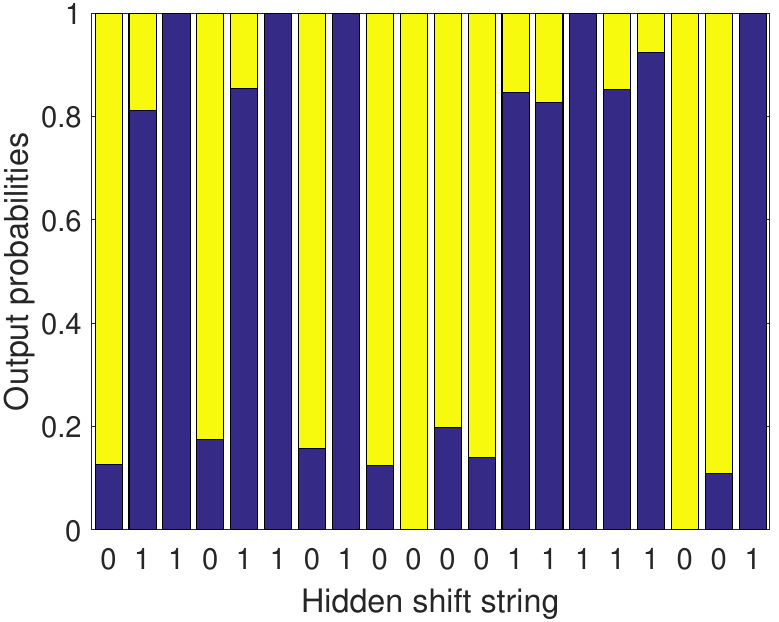}
\caption{Output single-qubit probability distributions
obtained by a classical simulation of the hidden shift quantum
algorithm on $n=40$ qubits. Only one half of all qubits are shown
(qubits $21,22,\ldots,40$). The final state of the algorithm is $|s\rangle=U|0^n\rangle$, where $s$ is
the hidden shift string to be found and $U$ is a Clifford+$T$ circuit
with the $T$-count $t=40$ (left) and $t=48$ (right).
In both cases the circuit $U$ contains a few hundred Clifford gates.
For each qubit the probability of measuring `1'  in the final state is indicated in blue. 
The $x$-axis labels indicate the correct hidden shift bits.
The entire simulation took several hours on  a  laptop computer.
}
\label{fig:plots}
\end{figure*}

We implemented our classical sampling algorithm in MATLAB and used it to simulate a class of benchmark quantum circuits 
on $n=40$  qubits, with a few hundred  Clifford gates,
and $T$-count $t\le 48$. Specifically, we simulated a quantum algorithm which solves the hidden shift problem for non-linear Boolean functions~\cite{Rotteler09}.
An instance of the hidden shift problem is defined by a pair of oracle functions
$f,f'\, : \, \FF_2^n\to \{\pm 1\}$ and a hidden shift  string $s\in \FF_2^n$. 
It is promised that $f$  is a  bent (maximally non-linear) function, that is, 
the Hadamard transform of $f$ takes values $\pm 1$.
It is also promised that $f'$ is the Hadamard transform of the shifted version of $f$, that is, 
\begin{equation}
f'(x)=2^{-n/2} \sum_{y\in \FF_2^n} (-1)^{x\cdot y}  f(y\oplus s) \quad \mbox{for all $x\in \FF_2^n$}.
\label{eq:hadamard}
\end{equation}
Here $\oplus$ stands for the bit-wise XOR.
The goal is to learn the hidden shift $s$ by making  as few queries to
$f$ and $f'$ as possible. The classical query complexity of this problem is
known to be linear in $n$, see Theorem~8 of Ref.~\cite{Rotteler09}.
In the quantum setting, $f$ and $f'$ are given as diagonal $n$-qubit unitary operators $O_f$ and $O_{f'}$
such that $O_f|x\rangle=f(x)|x\rangle$ and
$O_{f'}|x\rangle=f'(x)|x\rangle$ for all $x\in \FF_2^n$.
A quantum algorithm can learn $s$ by making 
a single query to each of these oracles, as can be seen from the identity~\cite{Rotteler09}
\begin{equation}
|s\rangle= U|0^n\rangle, \quad U\equiv H^{\otimes n} O_{f'} H^{\otimes n} O_{f}H^{\otimes n}.
\label{eq:Udescription}
\end{equation}
This hidden shift problem is ideally suited for our benchmarking task for two reasons. First, the algorithm produces a deterministic output, i.e., the output is a computational basis state $|s\rangle$ for some $n$-bit string $s$. Because of this we achieve the most favorable runtime scaling in Eq.~(\ref{T2})
since each bit of $s$  can be learned  by  calling the sampling algorithm with a single-qubit output register
($w=1$) and a constant statistical error $\epsilon$. 
Second,  the $T$-count of the algorithm can be easily controlled by 
choosing a suitable bent function. 
Indeed, the non-oracle part of the algorithm consists only of Hadamard gates. 
We  show that for a large class of bent functions  $f$ (from the so-called Maiorana-McFarland family)
the oracles $O_f$ and $O_{f'}$ 
can be constructed using Clifford gates and only a few $T$ gates, see the Supplemental Material for details.

The numerical simulations were performed for two randomly generated 
instances of the hidden shift problem with $n=40$ qubits. For each of these instances we simulated the quantum circuit for the hidden shift algorithm, i.e., the circuit implementing the unitary $U$ described above. 
The $T$-counts of the two simulated circuits are $t=40$ and $t=48$ respectively. Since the hidden shift $s$ is known beforehand, we are able to verify correctness of the simulation. Our results are presented in Fig.~\ref{fig:plots}. As one can see from the plots, the output probability distribution of each qubit has most of its weight at the corresponding value of the hidden shift bit. Only the output probabilities for qubits $21,22,\ldots,40$ are shown because our algorithm perfectly recovered the first half of the hidden shift bits $1,2,\ldots,20$. 
This perfect recovery occurs due to the special structure of the chosen bent functions,
see the Supplemental Material for further details.

The rest of the paper is organized as follows. 
In Section~\ref{sec:sketch}  we give an overview of our main techniques.  
In Section~\ref{sec:stabilizer} we summarize some
basic facts concerning stabilizer states. We present our classical simulation algorithms for Clifford+$T$ circuits 
in Section~\ref{sec:algorithms}. Finally, we show how to approximate tensor products
of magic states by linear combinations of stabilizer states in Section~\ref{sec:magic}.
In the Supplemental Material we provide pseudocode for the main subroutines used in our algorithms, and we discuss further details of the simulations reported in Fig.~\ref{fig:plots}.

\section{Sketch of techniques}
\label{sec:sketch}

Following Ref.~\cite{BSS15},
we simulate a Clifford+$T$ circuit classically
using three basic steps.  First, each $T$-gate in the original circuit 
is replaced by a certain gadget that contains only 
Clifford gates and a $0,1$-measurement. The Clifford gates
may be classically controlled by the measurement outcome. 
The gadget consumes one copy of the magic state $|A\rangle$. 
This gives an equivalent `gadgetized' circuit
acting on a non-stabilizer initial state that contains $t$ copies of $|A\rangle$. We show how to remove all intermediate measurements from the gadgetized circuit 
by replacing the outcomes of these measurements by random uniform postselection bits.
Accordingly, we  replace the classically controlled Clifford gates 
by a suitable random ensemble of uncontrolled  Clifford gates. 
Second, the initial  magic state $A^{\otimes t}$ is  represented (exactly or approximately) as 
 a linear combination of $\chi\ll 2^n$ stabilizer states. 
The action of the gadgetized circuit on each term in this linear combination can be 
efficiently simulated using the standard Gottesman-Knill theorem
since the gadgetized circuit contains only Clifford gates.
This allows us to represent the final state before the measurement of $Q_{out}$
as a linear combination of $\chi$ stabilizer states.
We simulate the measurement of $Q_{out}$ 
on this final state
independently for each term in the  linear combination
(we also have to simulate certain additional post-selective measurements  introduced at the 
first step).
This is possible due to the fact that $0,1$-measurements map stabilizer states
to stabilizer states. The final post-measurement state is  a linear combination of at most $\chi$
stabilizer states. 
The third and the most time consuming step is 
computing the norm of the post-measurement state. This norm
is simply related
to    the quantity of interest,
such as the output probability $P_{out}(x)$. 
We show how to obtain a square-root speedup in this step
compared with Ref.~\cite{BSS15}
reducing the runtime scaling from $\chi^2$  to $\chi$.
This is achieved using a novel subroutine for approximating the norm of a linear
combination of stabilizer states. The subroutine has runtime $O(\chi t^3 \epsilon^{-2})$,
where $\chi$ is the number of terms in the  linear combination,
$t$ is the number of qubits, and $\epsilon$ is the relative error. 
We expect that this subroutine may find applications in other contexts.
We achieve a further speedup compared with Ref.~\cite{BSS15}
by reducing the scaling $\chi\approx 2^{0.47 t}$ to
$\chi \approx 2^{0.23 t}$ by developing techniques for approximate
stabilizer decompositions of $A^{\otimes t}$. Although in general the simulation algorithm based on approximate stabilizer decomposition cannot accurately compute individual probabilities of the output distribution, we show that it can be used for sampling from the output distribution with a small statistical error.

\section{Stabilizer formalism}
\label{sec:stabilizer}

Before moving further, let us state some facts concerning stabilizer groups and stabilizer states. Let  $\calP_n$ be the $n$-qubit Pauli group.
Any element of $\calP_n$ has the form 
$i^m P_1\otimes \cdots \otimes P_n$, where
each factor $P_a$ is either the identity or a single-qubit
Pauli operator $X,Y,Z$ and $m\in \ZZ_4$. 
An abelian subgroup $\calG\subseteq \calP_n$ is called
a stabilizer group if $-I\notin \calG$. 
Each stabilizer group has the form $\calG=\langle G_1,\ldots,G_r\rangle$
for some generating set of pairwise commuting self-adjoint Pauli operators $G_1,\ldots,G_r\in \calG$ such that $|\calG|=2^r$.
The integer $r$ is called the dimension of $\calG$ and is denoted $r=\dim{(\calG)}$.
A state $\psi$ is said to be stabilized by $\calG$ if $P\psi=\psi$ for all $P\in \calG$.
States stabilized by $\calG$  span a ``codespace" of dimension $2^{n-r}$. 
A projector onto a codespace has the form
\begin{equation}
\label{PiG}
\Pi_\calG=2^{-r} \sum_{P\in \calG} P.
\end{equation}

A pure $n$-qubit state $\psi$ is a stabilizer state iff
$|\psi\rangle=U |0^{ n}\rangle$ for some Clifford unitary 
$U$.  Any stabilizer state $\psi$ is uniquely defined (up to the overall phase)
by a stabilizer group $\calG\subseteq \calP_n$ of dimension $n$ such that 
$\psi$ is the only state stabilized by $\calG$.
Let $\calS_n$ be the set of all $n$-qubit stabilizer states. This set is known to be a $2$-design~\cite{Dankert2009}, that is,
\begin{equation}
\label{2design}
|\calS_n|^{-1} \sum_{\psi\in \calS_n} |\psi\rangle\langle \psi|^{\otimes 2}
=\int d\mu(\phi) |\phi\rangle \langle \phi|^{\otimes 2},
\end{equation}
where the integral  is with respect to the Haar measure
on the set of all normalized $n$-qubit states $\phi$.

Throughout the paper we assume that stabilizer states are represented in a certain standard form
defined in Appendix~B. In this representation, three basic  tasks can be performed efficiently. First,
one can compute 
the inner product between  stabilizer states~\cite{aaronson2004improved,Garcia2014geometry,BSS15}.
More precisely, consider stabilizer states 
$\psi,\phi\in \calS_n$. Then
$\langle \psi|\phi\rangle =b 2^{-p/2} e^{i\pi m/4}$
for some $b=0,1$, integer $p\in [0,n]$ and $m\in \ZZ_8$ that can be computed in time $O(n^3)$,
see Ref.~\cite{BSS15}.
Pseudocode for computing the inner product $\langle \psi|\phi\rangle$
can be found in Appendix~C. Secondly, 
a projection of any stabilizer state onto the codespace
of any stabilizer code is a  stabilizer state which is easy to compute.
More precisely,  suppose $\calG\subseteq \calP_n$ is a stabilizer group
and $\varphi\in \calS_n$.
Then  $\Pi_\calG |\varphi\rangle = b 2^{-p/2} |\phi\rangle$
for some $b=0,1$, some integer $p\ge 0$, and
stabilizer state $\phi\in \calS_n$. One can compute
$b,p,\phi$ in time $O(rn^2)$ as explained in Appendix~E. Recall that $r=\dim{(\calG)}$.
Finally, one can generate a random stabilizer state drawn
from the uniform distribution on $\calS_n$ in time $O(n^2)$,
see Appendix~D.

\section{Classical simulation algorithms}
\label{sec:algorithms}

First consider the task of approximating
the output probability $P_{out}(x)$. 
The algorithm described below consists of two stages with runtimes 
\[
\tau_1=O\left( (w+t)(c+t)+(n+t)^3 \right) 
\]
and
\[
\tau_2=O( 2^{\beta t} t^3 \epsilon^{-2}\log(p_f^{-1})).
\]
The first stage computes  a stabilizer group $\calG\subseteq \calP_t$
and an integer $u$  such that 
\begin{equation}
\label{stage1a}
P_{out}(x)=2^{-u} \langle A^{\otimes t} |\Pi_\calG |A^{\otimes t}\rangle.
\end{equation} 
We begin by  replacing each $T$-gate in the original circuit $U$ 
by  the well-known gadget~\cite{Zhou2000} shown in Fig.~\ref{fig:Tgate}.
The gadget  implements the $T$-gate by performing  Clifford gates CNOT, $S$, and a $0,1$-measurement.
Each measurement outcome appears with the  probability $1/2$.
The gate $S$ is applied only if the outcome is '1'.
The gadget also consumes one copy of the magic state $|A\rangle$
which is destroyed in the process.

\begin{figure}[htbp]
\includegraphics[width=8cm]{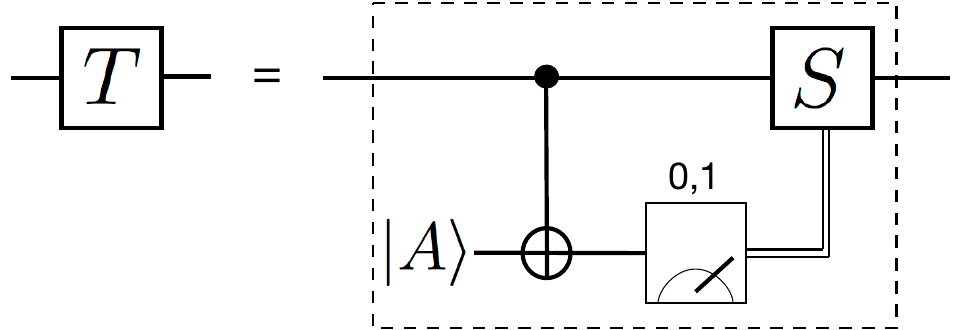}
\caption{The $T$-gate gadget.
The Clifford gate $S$ is classically controlled by the measurement outcome.
Both outcomes appear with probability $1/2$.
}
\label{fig:Tgate}
\end{figure}

Suppose we postselect the outcome '0' in each gadget, i.e. 
replace each measurement by a projector $|0\rangle\langle0|$.
This removes the classically controlled $S$-gates
such that each gadget adds a single CNOT to the original circuit $U$. 
Let $V$ be the modified version of $U$. By definition, $V$ acts on $n+t$ qubits and contains
$c+t$ Clifford gates. Let us agree that the $t$ ancillary qubits initialized in the magic state
are appended at the end of $n$ computational qubits such that the
circuit $V$ acts on the initial state $|0^{ n} A^{\otimes t}\rangle$. 
Combining the final measurement projector
$\Pi(x)=|x\rangle\langle x|_{Q_{out}} \otimes I_{else}$ with the projectors $|0\rangle\langle 0|$ acting on the ancillary
qubits  gives a projector 
\[
\Pi=\Pi(x)\otimes  |0^t\rangle\langle 0^t|
\]
acting on $n+t$ qubits 
such that 
\begin{equation}
\label{stage1b}
P_{out}(x)=2^t \langle 0^{n} A^{\otimes t} | V^\dag  \Pi V |0^{ n} A^{\otimes t}\rangle.
\end{equation}
Here we noted that the postselection probability is $2^{-t}$. 
Obviously, $\Pi=\Pi_\calW$ for a stabilizer group
$\calW\subseteq \calP_{n+t}$ of dimension $w+t$.
Namely, let $q(j)$ be the $j$-th  qubit of $Q_{out}$. 
Generators of $\calW$ are $R_j=(-1)^{x_j} Z_{q(j)}$
for $j=1,\ldots,w$ and 
$R_{w+j}=Z_{n+j}$ for $j=1,\ldots,t$. 
Since the conjugation by $V$ maps Pauli operators to Pauli operators,
we get 
$V^\dag \Pi_{\calW} V = \Pi_\calV$, where 
$\calV$ is a stabilizer group of dimension $w+t$
generated by 
$R_j'=V^\dag R_j V\in \calP_{n+t}$ with $j=1,\ldots,w+t$. Assuming that the action of a single Clifford gate
on a Pauli operator can be computed in time $O(1)$,
one can compute each generator
$R_j'$ in time $O(c+t)$. Accordingly, $\calV$ 
can be computed in  time $O((w+t)(c+t))$. 

Let $\calV_0$ be the subgroup of $\calV$ that includes all
Pauli operators which act as $I$ or $Z$ on each of the first $n$ qubits.
Let $v=\dim{(\calV_0)}$.  
A generating set $Q_1,\ldots,Q_v\in \calV_0$ can be computed in time 
$O(n(w+t) +(w+t)^3)=O((n+t)^3)$
using standard linear algebra. 
We get
\begin{equation}
\label{vpiv}
\langle 0^{n} |V^\dag \Pi V |0^{ n}\rangle=\langle 0^{n} | \Pi_\calV |0^{ n}\rangle=
2^{-w-t+v } \langle 0^{n} | \Pi_{\calV_0} |0^{n}\rangle.
\end{equation}
since $\langle 0^{ n}|P|0^{ n}\rangle=0$  $\; \forall \; P\in \calV\setminus \calV_0$.
Define $t$-qubit Pauli operators 
$G_i=\langle 0^{ n}|Q_i|0^{ n}\rangle$, $i=1,\ldots,v$.
These operators pairwise commute since $\calV_0$ is abelian 
and $Q_i$ commute with each other on the first $n$ qubits. 
If $-I\in \langle G_1,\ldots,G_v\rangle$ then there exists $Q\in \calV_0$ with $-I=\langle 0^{ n}|Q|0^{ n}\rangle$ and therefore
\[
\Pi_\calV |0^{ n}\rangle=\Pi_\calV Q|0^{ n}\rangle=-\Pi_\calV |0^{ n}\rangle=0
\]
in which case $P_{out}(x)=0$ and we are done.  Let us now consider the case $-I\notin \langle G_1,\ldots,G_v\rangle$. In this case let $\calG\subseteq \calP_t$ be the stabilizer group generated by $G_1,\ldots,G_v$
and $r=\dim{(\calG)}$.  One can check the condition $-I\notin \langle G_1,\ldots,G_v\rangle$ and compute $r$ in time $O(t^3)$.  Without loss of generality,
$\calG=\langle G_1,\ldots,G_r\rangle$. 
Noting that $\calV_0$ must contain $2^{v-r}$ elements acting trivially
on the last $t$ qubits yields
$\langle 0^{ n} |\Pi_{\calV_0} |0^{ n}\rangle = \Pi_\calG$.
This proves  Eq.~(\ref{stage1a}) with $u=w-v$ and the stabilizer group $\calG$  defined above.
Combining all the steps needed to compute $\calG$ gives the promised runtime
$\tau_1=O((w+t)(c+t)+(n+t)^3)$.

The second stage of the algorithm computes the expectation value in Eq.~(\ref{stage1a})
by decomposing $|A^{\otimes t}\rangle$ into a linear combination 
of stabilizer states. 
Suppose 
\begin{equation}
\label{Aexact}
|A^{\otimes t}\rangle = \sum_{a=1}^\chi y_a |\varphi_a\rangle
\end{equation}
for some stabilizer states $\varphi_a\in \calS_t$ and some coefficients $y_a$. 
For each $a=1,\ldots,\chi$ compute $b_a\in \{0,1\}$, an integer $p_a\ge 0$ and a stabilizer
state $\phi_a\in \calS_t$ such that 
\[
\Pi_\calG |\varphi_a\rangle=b_a 2^{-p_a/2}  |\phi_a\rangle.
\]
see Appendix~E for details. As stated above, 
this computation takes time $O(\chi t^3)$. 
Introducing new coefficients $z_a=2^{-(u+p_a)/2}y_a b_a$ 
and using Eqs.~(\ref{stage1a},\ref{Aexact}) one gets 
\begin{equation}
\label{EQ4}
P_{out}(x)=\|\psi\|^2, \quad |\psi\rangle =\sum_{a=1}^\chi z_a |\phi_a\rangle,
\quad \phi_a\in \calS_{t}.
\end{equation}
Here $\phi_a$
are $t$-qubit stabilizer states. Below we describe a randomized algorithm that takes as input a $t$-qubit state $\psi$, a target error parameter $\epsilon>0$ and a failure probability $p_{f}$. The algorithm computes a real number $\xi$ which, with probability at least $1-p_{f}$, approximates the norm of $\psi$ with relative error $\epsilon$. The running time of the algorithm is $O(\chi t^3 \epsilon^{-2}\log(p_f^{-1}))$. The key idea is to approximate $\|\psi\|^2$ by computing inner products
between $\psi$ and randomly chosen stabilizer states.  

We shall first consider the special case where the failure probability is $1/4$; at the end we describe how to reduce it to a given value $p_f^{-1}$. Let $\theta\in \calS_t$ be a random stabilizer state drawn from the uniform distribution.
Define expectation values 
\[
M_2\equiv \EE_\theta |\langle \theta|\psi\rangle|^2 \quad \mbox{and} \quad
M_4\equiv \EE_\theta |\langle \theta|\psi\rangle|^4.
\]
Using Eq.~(\ref{2design}) one can compute $M_2$ and $M_4$ by 
pretending that $\theta$ is drawn from the Haar measure.  Standard formulas
for the  integrals over the unit sphere yield
\begin{equation}
\label{M2M4}
M_2=\frac{\|\psi\|^2}{d} \quad \mbox{and} \quad M_4=\frac{2 \|\psi\|^4}{d(d+1)}, 
\quad \mbox{where} \quad d\equiv 2^t.
\end{equation}
Suppose $\theta_1,\ldots,\theta_L\in \calS_t$ are random independent
stabilizer states. 
Define a random variable 
\begin{equation}
\label{xi_main}
\xi=\frac{d}L \sum_{i=1}^L |\langle \theta_i|\psi\rangle|^2.
\end{equation}
From Eq.~(\ref{M2M4}) one infers that the expected value of $\xi$ is $\bar{\xi}=\EE(\xi)=\| \psi \|^2$ and 
the standard deviation of $\xi$ is 
\[
\sigma=\sqrt{ d^2 L^{-1} (M_4-M_2^2)} = \sqrt{\frac{d-1}{d+1}} L^{-1/2} \|\phi\|^2.
\]
For large $t$ one has $\sigma \approx L^{-1/2} \|\psi\|^2$. By the Chebyshev inequality, $\prob{|\xi-\bar{\xi}|\ge 2\sigma}\le \frac{1}{4}$.
Thus 
\begin{equation}
(1-\epsilon)\|\psi\|^2 \le \xi \le (1+\epsilon)\|\psi\|^2
\label{eq:errorbound}
\end{equation}
with probability at least $3/4$ provided that 
$L=4\epsilon^{-2}$. 

Now let us discuss how to reduce the failure probability (from $1/4$) so that it is below a given value $p_f^{-1}$. To achieve this, we compute independent estimates $\xi_1,\xi_2,\ldots \xi_J$ using the above procedure and output the median $\xi_{\mathrm{med}}$ of these values. It is a simple fact that this procedure reduces the failure probability to below $p_f^{-1}$ using only $J=O(\log(p_f^{-1}))$ estimates (see Lemma 6.1 of Ref.~\cite{jerrum1986random}). With this choice, the probability that Eq.~\eqref{eq:errorbound}  holds with $\xi$ replaced by $\xi_{\mathrm{med}}$ is at least $1-p_f$.

The inner product
$\langle \theta_i|\psi\rangle=\sum_{a=1}^\chi z_a \langle \theta_i|\phi_a\rangle$
in Eq.~(\ref{xi_main})
can be computed in time $O(\chi t^3)$ since $\theta_i$ and $\phi_a$
are stabilizer states of $t$ qubits. It follows that $P_{out}(x)=\|\psi\|^2$ can be approximated in time $O(\chi t^3 \epsilon^{-2}\log(p_f^{-1}))$, as promised. 

Since the runtime grows linearly with $\chi$, we would like 
to choose  a stabilizer decomposition in Eq.~(\ref{Aexact}) with 
a small rank $\chi$.
 Clearly, the optimal choice is $\chi=\chi_t$, where
$\chi_t\equiv \chi_t(0)$ is the stabilizer rank defined in the introduction.
Unfortunately, the exact value of $\chi_t$ is unknown.
Using the identity 
\begin{equation}
\label{AA}
|A^{\otimes 2}\rangle =\frac12 (|00\rangle + i|11\rangle) + \frac{e^{i\pi/4}}2 (|01\rangle+|10\rangle)
\end{equation}
one can see that $A^{\otimes 2}$ is a linear combination of two stabilizer states, that is,
$\chi_2=2$. By dividing $t$ qubits into $t/2$ pairs 
and applying the decomposition Eq.~(\ref{AA}) to each pair 
one gets $\chi_t\le 2^{t/2}$.  The results of~\cite{BSS15} 
give a slightly better bound $\chi_t\le 2^{\beta t}$ with $\beta\approx 0.47$.
This completes the analysis of the first algorithm.

{\em Remark 1:}
If $\calG$ has a small dimension, namely,
$r<\beta$, 
it can be  easier to compute $P_{out}(x)$ directly from 
Eqs.~(\ref{PiG},\ref{stage1a})
which yield $P_{out}(x)\sim \sum_{P\in \calG} \langle A^{\otimes t}|P|A^{\otimes t}\rangle$.
Clearly, each term in the sum can be computed in time $O(t)$,
so the overall runtime becomes $O(t|\calG|)=O(t 2^r)$.

{\em Remark~2:}
An alternative strategy to estimate the
expectation value in Eq.~(\ref{stage1a})
is to compute the inner products
\[
\langle \Pi_\calG \theta_i | A^{\otimes t}\rangle =\sum_{a=1}^\chi y_a \langle  \Pi_\calG\theta_i |\varphi_a\rangle
\]
for $i=1,\ldots,L$. Here $\varphi_a$ are the stabilizer states defined in Eq.~(\ref{Aexact})
and $\theta_i$ are random stabilizer states.
The same arguments as above show that 
\[
\langle A^{\otimes t} |\Pi_\calG|A^{\otimes t}\rangle =\| \Pi_\calG A^{\otimes t} \|^2
\approx \frac{d}{L} \sum_{i=1}^L |\langle \Pi_\calG \theta_i | A^{\otimes t}\rangle|^2.
\]
This  may be beneficial in the regime $L\ll \chi$ since one
has to compute the action of $\Pi_\calG$ only $L$ times rather than
$\chi$ times. 

Let us now describe the  algorithm that allows one to sample
$x$ from the distribution $P_{out}$ with statistical error $\epsilon$.  As before, we replace each $T$-gate in the original circuit $U$
by the gadget shown on Fig.~\ref{fig:Tgate},
prepare all magic states $|A\rangle$ at the very first time step,
and permute the qubits such that the initial state is $|0^n A^{\otimes t}\rangle$.
Let $y_j\in \{0,1\}$ be the outcome of the measurement performed
in the $j$-th gadget and $y=(y_1,\ldots,y_t)$.
Let $V_y$ be the Clifford circuit on $n+t$ qubits corresponding to measurement outcomes $y$. Each gadget with $y_j=0$ contributes a CNOT gate to $V_y$,
whereas each gadget with $y_j=1$ contributes a CNOT and the $S$-gate to $V_y$.
Thus $V_y$ contains $c+t+|y|$ gates.
A composition of all gadgets and Clifford gates of $U$ implements a
 trace preserving completely positive (TPCP) map
\[
\Phi(\rho)=\sum_y  (I_n \otimes |y\rangle\langle y|) V_y \rho V_y^\dag  (I_n \otimes |y\rangle\langle y|).
\]
Here $I_n$ is the $n$-qubit identity operator
and the sum runs over all $t$-bit strings $y$.
Suppose first that $\Phi$ is applied to a state
$\rho_{in}=|0^n\rangle\langle 0^n|\otimes |A\rangle\langle A|^{\otimes t}$.
Then the final state of the $n$  computational qubits is $U|0^n\rangle$
regardless of $y$ and each $y$ appears with
probability $2^{-t}$. Thus 
\begin{equation}
\label{PHIa}
\Phi(\rho_{in})=U|0^n\rangle\langle 0^n|U^\dag \otimes \frac{I}{2^t}.
\end{equation}
Next suppose that $\Phi$ is applied to a state
$\tilde{\rho}_{in}=|0^n\rangle\langle 0^n|\otimes |\psi\rangle\langle \psi|$,
where $\psi$ is a  linear combination of $\chi$ stabilizer states
$\varphi_1,\ldots,\varphi_\chi\in \calS_t$ that 
approximates $A^{\otimes t}$ with a small error:
\begin{equation}
\label{Aprox}
|\psi\rangle = \sum_{a=1}^\chi z_a |\varphi_a\rangle,
\quad  \quad
|\langle A^{\otimes t} | \psi\rangle|^2 \ge 1-\epsilon^2/25.
\end{equation}
Here $z_a$ are some coefficients and 
we assume $\psi$ has unit norm. The error $\epsilon^2/25$ is sufficient to ensure that the output distribution of the overall simulation algorithm is $\epsilon$-close to  $P_{out}$.
From Eq.~(\ref{Aprox}) one gets
\begin{equation}
\label{error1}
\|\rho_{in}-\tilde{\rho}_{in}\|_1 = \| \, |A\rangle\langle A|^{\otimes t} - |\psi\rangle\langle \psi| \, \|_1  \le \frac{2}{5}\epsilon.
\end{equation}
By definition of $\Phi$, 
\begin{equation}
\label{PHIb}
\Phi(\tilde{\rho}_{in})=\sum_y p_y\,  |\phi_y\rangle\langle \phi_y| \otimes |y\rangle\langle y|,
\end{equation}
where
\begin{equation}
\label{p_y}
p_y=\langle 0^n \otimes \psi | V_y^\dag (I_n \otimes |y\rangle\langle y| ) V_y |0^n \otimes  \psi \rangle
\end{equation}
and $\phi_y$ are normalized $t$-qubit states defined by 
\begin{equation}
\label{phi_y}
|\phi_y\rangle=p_y^{-1/2} \langle y| V_y|0^n \otimes  \psi\rangle.
\end{equation}
Clearly, $p$ is a normalized probability distribution on the set of $t$-bit strings.
The state $\phi_y$ is defined only for $p_y>0$.
Combining Eqs.~(\ref{PHIa},\ref{error1},\ref{PHIb}) and tracing out the last $t$ qubits 
of $\Phi(\rho_{in})$ and $\Phi(\tilde{\rho}_{in})$ 
one infers that 
\begin{equation}
\label{error2}
\|\,  U|0^n\rangle\langle 0^n|U^\dag - \sum_y p_y |\phi_y\rangle\langle \phi_y| \, \|_1 \le \frac{2}{5}\epsilon.
\end{equation}
Here we noted that TPCP maps do not increase the trace distance. Combining Eqs.~(\ref{PHIa},\ref{error1},\ref{PHIb}) and tracing out the first $n$ qubits of $\Phi(\rho_{in})$ and $\Phi(\tilde{\rho}_{in})$ shows that the distribution $p$ satisfies $\|p-u\|_1\leq \frac{2}{5}\epsilon$, where $u$ is the uniform distribution on the set of $t$-bit strings. Using this fact and Eq.~(\ref{error2}) we arrive at
\begin{equation}
\label{error3}
\|\,  U|0^n\rangle\langle 0^n|U^\dag -  \frac{1}{2^t}\sum_y|\phi_y\rangle\langle \phi_y| \, \|_1 \le \frac{4}{5}\epsilon.
\end{equation}

For each $t$-bit string $y$ define a probability distribution ${P}^{y}_{out}(x)=\langle \phi_y|\Pi(x)|\phi_y\rangle$. 
Below we give an algorithm which takes as input $y$ and $\epsilon$ and produces a sample from a distribution $\tilde{P}^{y}_{out}$ which satisfies 
\begin{equation}
\label{Rdist}
\|{P}^{y}_{out}(x)-\tilde{P}^{y}_{out}(x)\|_1\leq \epsilon/5
\end{equation}
Our algorithm to approximately sample from $P_{out}$ has two steps. We first generate a random uniformly distributed  $t$-bit string $y$ and then we sample $x$ from $\tilde{P}^{y}_{out}$. From Eqs.~(\ref{error3},\ref{Rdist}) we see that the distribution over outputs $x\in \{0,1\}^w$ produced by this algorithm approximates $P_{out}$ within error $\epsilon$ in the trace norm.  

We are now ready to describe how to sample from $\tilde{P}_{out}^y$ satisfying Eq.~(\ref{Rdist}). We first describe how to compute an approximation to $P_{out}^y(x)$ with relative error $\delta$.  Note that
\begin{equation}
\label{tildePa}
P^{y}_{out}(x) = \frac{ \langle 0^n \otimes \psi| V_y^\dag (\Pi(x) \otimes |y\rangle\langle y|) V_y |0^n \otimes \psi\rangle}
{ \langle 0^n\otimes  \psi| V_y^\dag (I_n \otimes |y\rangle\langle y|) V_y |0^n\otimes \psi\rangle }.
\end{equation}
Here we used Eqs.~(\ref{p_y},\ref{phi_y}).
Repeating the same arguments as in the derivation of Eq.~(\ref{stage1a}) one gets
\begin{equation}
\label{tildePb}
P^y_{out}(x)=\frac{ 2^{-u} \langle \psi| \Pi_\calG|\psi\rangle }
{2^{-v} \langle \psi| \Pi_\calH|\psi\rangle }
\end{equation}
for some stabilizer groups $\calG,\calH\subseteq \calP_t$
and integers $u,v$ that can be computed in time
$\tau_1=O((w+t)(c+t) + (n+t)^3)$. We already know a randomized algorithm which computes $\langle \psi| \Pi_\calG|\psi\rangle$
and $\langle \psi| \Pi_\calH|\psi\rangle$
with a relative error $\delta$ in time $\tau_2=O(\chi t^3 \delta^{-2}\log(p_f^{-1}))$. Recall that $p_f$ is the probability that the algorithm does not achieve the desired approximation. Thus we can compute $P^y_{out}(x)$ 
with a relative error $2\delta$ in time $\tau_1+\tau_2$.

Now consider the task of sampling from $P^y_{out}$. Assume for simplicity that $Q_{out}=\{1,2,\ldots,w\}$.
For each $j=1,\ldots,w-1$ define conditional probabilities
\begin{equation}
\label{conditional}
P^y_{out}(z | x_1,\ldots,x_{j-1}) =\frac{P^y_{out}(x_1,\ldots,x_{j-1},z)}{P^y_{out}(x_1,\ldots,x_{j-1})},
\end{equation}
where $z\in\{0,1\}$. Suppose the bits $x_1,\ldots,x_{j-1}$  have already been sampled (initially $j=1$). Then the next bit $x_j$ can be sampled by tossing a coin with bias $P^y_{out}(0 | x_1,\ldots,x_{j-1})$. Things are complicated by the fact that we cannot exactly compute this conditional probability.  We use the same simulation strategy except that at each step the conditional probability $P^y_{out}(0 | x_1,\ldots,x_{j-1})$ is replaced by an approximation $q_j$. Here we require that with probability at least $1-p_f$,  both $q_j$ and $1-q_j$ approximate the conditional probabilities $P^y_{out}(0| x_1,\ldots,x_{j-1})$ and $P^y_{out}(1| x_1,\ldots,x_{j-1})$ respectively with relative error $O(\delta)$. Such an approximation $q_j$ can be computed in time $O(\tau_1+\tau_2)$ using the procedure described above for approximating the probabilities on the right-hand side of Eq.~(\ref{conditional}). Indeed, 
first compute $a,b$ which, with probability at least $1-p_f$, approximate $P^y_{out}(0| x_1,\ldots,x_{j-1})$ and
 $P^y_{out}(1| x_1,\ldots,x_{j-1})$ respectively with relative error $\delta$. If $a\leq b$ then we set $q_j=a$ while if $b<a$ then we set $q_j=1-b$.
 
 We now analyze the resulting simulation algorithm and show that we can ensure Eq.~(\ref{Rdist}) by choosing approximation error $\delta=O(\epsilon w^{-1})$ and failure probability $p_f=O(\epsilon w^{-1})$. Let us first suppose that all probabilities $q_j$ computed by the algorithm achieve the desired approximation $\delta$ (i.e., no failures occur). Conditioned on this event we see that the output distribution produced by the algorithm approximates $P_{out}^y(x)$ with relative error $O(\delta w)$. This conditional probability distribution can therefore be made $\epsilon/10$-close (say) to $P_{out}^y$ by choosing $\delta=O(\epsilon w^{-1})$. It remains to show that by choosing $p_f=O(\epsilon w^{-1})$ we can ensure that the output distribution $\tilde{P}^y_{out}$ of the simulation algorithm is $\epsilon/10$-close to the distribution conditioned on no failures. This follows because the algorithm computes $O(w)$ probabilities $\{q_j\}$ in total and choosing $p_f=O(\epsilon w^{-1})$ we can ensure that all of them are computed to within the desired approximation error $\delta$, with probability at least $1-\epsilon/20$. With this choice we have $\tilde{P}^y_{out}=(1-\epsilon/20)P_{A}+\epsilon/20 P_{B}$ where $P_{A}$ is the distribution conditioned on no failures, and thus  $\|\tilde{P}^y_{out}-P_{A}\|_1\leq \epsilon/10$ as claimed.

The overall running time of this algorithm is $\tau_1'+\tau_2'$, where 
$\tau_1'=O(w\tau_1) =O(w(w+t)(c+t) + w(n+t)^3)$
and $\tau_2'=O(w\tau_2) = O(\chi w^3  t^3\epsilon^{-2}\log(w\epsilon^{-1}))$.

{\em Remark:} This algorithm can be modified slightly to handle certain Clifford+$T$ circuits which use measurement and classical control. To see how, recall that in the $T$-gate gadget from Fig.~\ref{fig:Tgate}, a single qubit is measured in the computational basis (yielding both outcomes with equal probability) and a Clifford operation is classically controlled on the measurement outcome. In our simulation algorithm the measurement is replaced by a uniformly chosen postselection bit $y_j$. Exactly the same strategy can be used for other simple gadgets which involve measurement and classical control.  For example, the Toffoli gate can be implemented as a Clifford+$T$ circuit with only four $T$-gates if we allow two ancillas, measurement, and classical control \cite{Jones2013}(otherwise it requires seven $T$-gates \cite{amy2013meet,GKMR2014}). Fortunately it is possible to use the less costly circuit with four $T$-gates in the above simulation algorithm by including one additional postselection bit per Toffoli gate.

\section{Approximating magic states}
\label{sec:magic}

In this section we  show how to  compute a decomposition Eq.~(\ref{Aprox}) with $\chi=O(2^{\gamma t} \epsilon^{-2})$, where $\gamma$ satisfies Eq.~(\ref{T3}).
Define a state
\[
|H\rangle=\cos(\pi/8)|0\rangle+\sin(\pi/8)|1\rangle.
\]
We note that the magic state $|A\rangle$ is equivalent to $|H\rangle$
modulo Clifford gates and a global phase,  $|A\rangle=e^{i\pi/8} H S^\dagger |H\rangle$. Thus
it suffices to construct a state 
\begin{equation}
\label{AproxT}
|\psi\rangle = \sum_{a=1}^\chi z_a |\varphi_a\rangle, \quad \varphi_1,\ldots,\varphi_\chi\in \calS_t
\end{equation}
such that $\|\psi\|=1$, 
\begin{equation}
\label{AproxT1}
|\langle H^{\otimes t} | \psi\rangle|^2 \ge 1-\delta \quad \mbox{and} \quad 
\chi=O(2^{\gamma t} \delta^{-1})
\end{equation}
for all sufficiently small $\delta>0$.

Our starting point is the identity 
\begin{equation}
\label{eq:Texact}
|H^{\otimes t}\rangle=\frac{1}{(2\nu)^t}\sum_{x\in \mathbb{F}_2^t} |\tilde{x}_1\otimes \tilde{x}_2\otimes\ldots \otimes \tilde{x}_t\rangle
\end{equation}
where $|\tilde{0}\rangle\equiv |0\rangle$, $|\tilde{1}\rangle\equiv H|0\rangle=2^{-1/2}(|0\rangle+|1\rangle)$, and 
\[
\nu \equiv \cos(\pi/8).
\] 
The right-hand side of Eq.~(\ref{eq:Texact}) is a uniform superposition of $2^t$ non-orthogonal stabilizer states labeled by elements of the vector space $\mathbb{F}_2^t$. We shall construct an approximation
$\psi$ which is a uniform superposition
of states $|\tilde{x}_1\otimes \tilde{x}_2\otimes\ldots \otimes \tilde{x}_t\rangle$
 over a linear subspace of $\FF_2^t$.

Let $L(t,k)$ be the set of all $k$-dimensional linear subspaces $\calL\subseteq \FF_2^t$. We will fix $k$ below.  For each $\calL\in L(t,k)$ define a  state 
\begin{equation}
\label{eq:Lsubspace}
|\calL\rangle=\frac{1}{\sqrt{2^k Z(\calL)}} \sum_{x\in \calL} |\tilde{x}_1\otimes \tilde{x}_2 \otimes \cdots \otimes \tilde{x}_t\rangle 
\end{equation}
where 
\begin{equation}
\label{eq3}
Z(\calL)\equiv \sum_{x\in \calL} 2^{-|x|/2}.
\end{equation}
Using the identity $\langle \tilde{a}|\tilde{b}\rangle=2^{-|a\oplus b|/2}$,
where $a,b\in \{0,1\}$, 
and the fact that $\calL$ is a linear subspace one can easily check that $|\calL\rangle$
is a normalized state, $\langle \calL|\calL\rangle=1$.
We take our approximation $\psi$ from Eq.~(\ref{AproxT}) to be Eq.~(\ref{eq:Lsubspace}) for a suitably chosen 
subspace $\calL^\star\in L(t,k)$, which gives an approximate decomposition
of $H^{\otimes t}$ using $\chi=2^k$ stabilizer states. 
How small can we hope to make $k$? Using the fact that $\langle H|\tilde{0}\rangle=\langle H|\tilde{1}\rangle=\nu$ we see that
\begin{equation}
\label{eq:fidT}
|\langle H^{\otimes t}|\calL\rangle|^2=\frac{2^k \nu^{2t}}{Z(\calL)}
\end{equation}
From this we immediately get a lower bound on $k$. Indeed, since $Z(\calL)\geq 1$ we will need $2^k\geq  \nu^{-2t}(1-\delta)$ to achieve the desired approximation. 
Below we describe a randomized algorithm which outputs a subspace $\calL^{\star}$ with $2^k=O(\delta^{-1}\nu^{-2t})$. Thus for constant $\delta$ we achieve the best possible scaling of $k$ with $t$. We will use the following fact about random subspaces of $\mathbb{F}_2^t$.

\begin{lemma}
Let $\calL\in L(t,k)$ be chosen uniformly at random. Then 
\begin{equation}
\label{eq4}
\EE(Z(\calL)) \le  1+2^k\nu^{2t}.
\end{equation}
\label{lem:expectation}
\end{lemma}
\begin{proof}
By linearity, we have 
\begin{equation}
\label{eq5}
\EE(Z(\calL)) = 1+\sum_{x\in \FF_2^t \setminus 0}  2^{-|x|/2}\cdot \EE(\chi_\calL(x)),
\end{equation}
where $\chi_\calL(x)$ is the indicator function of $\calL$. The expectation value $\EE(\chi_\calL(x))$ with respect to $\calL$ for a fixed $x$
is $(2^k-1)/(2^t-1)$. Thus we arrive at
\begin{align*}
\label{eq6}
\EE(Z(\calL)) &= 1+\frac{(2^k-1)}{(2^t-1)} \sum_{x\in \FF_2^t \setminus 0}  2^{-|x|/2}\\
& = 1+\frac{(2^k-1)}{(2^t-1)} \left( 2^t\nu^{2t} -1\right)\\
&\leq 1+2^k\nu^{2t}.
\end{align*}
\end{proof}
As a corollary, there exists at least one $\calL\in L(t,k)$ such that 
$Z(\calL)\le 1+2^k\nu^{2t}$.
We now fix $k$ to be the unique positive integer satisfying
\begin{equation}
\label{eq:kstar}
4\geq  2^{k}\nu^{2t}\delta \geq 2.
\end{equation}
Consider a subspace $\calL\in L(t,k)$ chosen uniformly at random. Using Markov's inequality and Lemma \ref{lem:expectation} we get
\begin{align}
\mathrm{Pr}\left[\frac{Z(\calL)}{(1+2^{k}\nu^{2t} )(1+\delta/2)}\geq 1\right]& \leq \frac{\EE(Z(\calL))}{(1+2^{k}\nu^{2t} )\left(1+\delta/2\right)}\nonumber\\
& \leq 1-\frac{\delta}{2+\delta}\nonumber.
\end{align}
For a given $\calL\in L(t,k)$ we may compute $Z(\calL)$ in time $O(2^{k})$. By randomly choosing $O(1/\delta)$ subspaces $\calL$ we obtain one $\calL^{\star}$ satisfying 
\begin{equation}
\label{eq:Zstar}
Z(\calL^{\star})\leq (1+2^{k}\nu^{2t})(1+\delta/2)
\end{equation}
with constant probability. Plugging Eq.~(\ref{eq:Zstar}) into Eq.~(\ref{eq:fidT}) we see that
\begin{align*}
|\langle H^{\otimes t}|\calL^\star\rangle|^2 & \geq \frac{1}{\left(1+2^{-k}\nu^{-2t}\right)(1+\delta/2)} \nonumber\\ & \geq \frac{1}{(1+\delta/2)^2}\nonumber\\&\geq1-\delta,
\end{align*}
where in the second line we used Eq.~(\ref{eq:kstar}). The state $|\psi\rangle=|\calL^\star\rangle$ obtained in this way therefore satisfies Eq.~(\ref{AproxT1}) with 
\begin{equation}
\label{chi_upper_bound}
\chi=2^k\leq 4\nu^{-2t}\delta^{-1}=O(\nu^{-2t}\delta^{-1})=O(2^{\gamma t} \delta^{-1}).
\end{equation}
This algorithm has running time $O(\nu^{-2t} \delta^{-2})$, since we must check the condition Eq.~(\ref{eq:Zstar}) for each of the $O(\delta^{-1})$ randomly sampled elements of $L(t,k)$ (note that the time required to sample each element is $O(poly(t))$).

{\em Remark:} 
One may ask whether a stronger bound on $\chi$ can be obtained by 
truncating the expansion of $H^{\otimes t}$ in some other basis
of stabilizer states. For example, consider the standard $0,1$-basis
of $t$ qubits. The expansion of  $H^{\otimes t}$  in this basis
is concentrated on basis vectors $x\in \FF_2^t$ with Hamming weight
$|x|=(1-\nu^2)t \pm O(t^{1/2})$. The number of such basis vectors
scales as $\chi\sim 2^{ t H_2(\nu^2)}\approx 2^{0.6 t}$, where $H_2(p)$ is the binary Shannon
entropy function. Thus replacing the $\tilde{0},\tilde{1}$-basis
by the $0,1$-basis gives a significantly worse bound on $\chi$.

As noted above, taking $\delta$ to be a constant our construction has the best possible scaling  $\chi=O(\nu^{-2t})$ of any decomposition of the form Eq.~(\ref{eq:Lsubspace}). 
In fact, we  prove the following lower bound on the stabilizer rank of $H^{\otimes t}$.
\begin{lemma}
\label{lemma:lower}
Consider a state $|\psi\rangle=\sum_{a=1}^\chi z_a |\phi_a\rangle$, where
$\phi_a\in \calS_t$. Suppose $\|\psi\|=1$ and 
$|\langle \psi| H^{\otimes t}\rangle| \ge f$.
Then $\chi\ge \nu^{-2t}  f^2 \| z\|^{-2}$, 
where  $z=(z_1,\ldots,z_\chi)\in \CC^\chi$.
\end{lemma}
\begin{proof}
First, let us show that 
\begin{equation}
\label{max1}
F_t\equiv \max_{\phi\in \calS_t} |\langle \phi|H^{\otimes t}\rangle|=\nu^t.
\end{equation}
The lower bound $F_t\ge \nu^t$ is obvious since $\langle 0^{\otimes t}|H^{\otimes t}\rangle=\nu^t$.
We shall use induction in $t$ to show that $F_t\le \nu F_{t-1}$.
Consider some fixed $t$ and let 
$F_t=|\langle \phi|H^{\otimes t}\rangle|$ for some $\phi\in \calS_t$.
Suppose we measure the first qubit of $\phi$ in the $0,1$ basis.
Let $P_a$ be the probability of getting the outcome $a=0,1$. 
It is well-known that $P_a\in \{0,1,1/2\}$ for any stabilizer state $\phi$.
Consider three cases.

\noindent
{\em Case 1:} $P_0=1$. Then $|\phi\rangle=|0\rangle \otimes |\psi\rangle$
for some $\psi\in \calS_{t-1}$ and $F_t=\nu |\langle \psi| H^{\otimes (t-1)}\rangle| \le \nu F_{t-1}$.

\noindent
{\em Case 2:} $P_0=0$. Then $|\phi\rangle=|1\rangle \otimes |\psi\rangle$
for some $\psi\in \calS_{t-1}$ and $F_t=\sqrt{1-\nu^2} |\langle \psi| H^{\otimes (t-1)} \rangle| <\nu F_{t-1}$.

\noindent
{\em Case 3:} $P_0=1/2$. Then 
\[
|\phi\rangle = 2^{-1/2} \left( |0\rangle \otimes |\psi_0\rangle + |1\rangle \otimes |\psi_1\rangle \right)
\]
for some $\psi_0,\psi_1\in \calS_{t-1}$. By triangle inequality,
\[
F_t\le 2^{-1/2} (\nu+\sqrt{1-\nu^2}) F_{t-1} =\nu F_{t-1}.
\]
The base of induction $F_1=\nu$ is trivial. 
This proves Eq.~(\ref{max1}).
From Eq.~(\ref{max1}) one gets
\[
f\le |\langle \psi |H^{\otimes t}\rangle| \le \nu^t \sum_{a=1}^\chi |z_a| \le \nu^t \chi^{1/2} \|z\|.
\]
This is equivalent to the statement of the lemma.
\end{proof}
We conjecture that  any  approximate
stabilizer decomposition of $H^{\otimes t}$
that achieves a constant approximation error must use at least
$\Omega(\nu^{-2t})$ stabilizer states.

\section{Acknowledgments}
DG acknowledges funding provided by the Institute for Quantum Information and Matter, an NSF Physics Frontiers Center (NFS Grant PHY-1125565) with support of the Gordon and Betty Moore Foundation (GBMF-12500028). 
SB thanks Alexei Kitaev for helpful discussions and comments.

\section*{Appendix A: Quadratic forms}
\label{app:A}

The remaining sections  provide more
details on implementation of our algorithms.
 Appendix~A presents some basic facts 
about quadratic forms over finite fields
and describes a subroutine for computing
certain exponential sums.
The standard form of stabilizer states used in all our
algorithms is defined in Appendix~B.
Then we present algorithms for computing the inner product between stabilizer states
(Appendix~C), generating  a random uniformly distributed stabilizer state (Appendix~D),
and computing the action
of Pauli measurements on stabilizer states (Appendix~E).
The three algorithms have running time $O(n^3)$, $O(n^2)$,
and $O(n^2)$ respectively, where $n$ is the number of qubits.
We provide pseudocode for all algorithms and report timing analysis
for a MATLAB implementation. Appendix~F describes simulation of the hidden shift algorithm. 

Below we consider functions that map binary vectors
to integers modulo eight. We define a 
special class of such functions that are analogous to quadratic forms
over the real field. 
The definition of $\ZZ_8$-valued quadratic forms given below
was proposed to us by Kitaev~\cite{Kitaev2003}. Analogous definitions and computations using $\ZZ_4$-valued quadratic forms can be found in \cite{schmidt2009}. For a general theory of quadratic forms over a finite field see Ref.~\cite{Araujo2011}. 
Throughout the rest of the paper arithmetic operations $\pm$ are performed modulo eight
(unless stated otherwise),
whereas addition of binary vectors modulo two is denoted $\oplus$.
Elements of  $\FF_2^n$ are considered as binary row vectors.
A binary inner product between vectors $x,y\in \FF_2^n$ will be denoted
$(x,y)\equiv \sum_{i=1}^n x_i y_i {\pmod 2}$.
A set of binary matrices of size $a\times b$ is denoted $\FF_2^{a\times b}$.
A transpose of a matrix $M$ is denoted $M^T$.

Recall that a subset  $\calK\subseteq \FF_2^n$ is a 
called an affine space of dimension $k$ iff
$\calK=\calL(\calK)\oplus h$ for some 
$k$-dimensional linear subspace $\calL(\calK)\subseteq \FF_2^n$
and a shift vector $h\in \FF_2^n$.
Note that $\calK$ uniquely determines $\calL(\calK)$, namely, $\calL(\calK)=\{ x\oplus y\, : \, x,y\in \calK\}$. 
The shift vector $h$ however is not uniquely defined.
Obviously, $|\calK|=|\calL(\calK)|=2^k$.
\begin{dfn}
\label{dfn:qform}
Consider an affine space $\calK\subseteq \FF_2^n$.
A function $q\, : \, \calK \to \ZZ_8$ is called a quadratic form iff
there exists a function $J\, : \, \calL(\calK)\times \calL(\calK) \to \ZZ_8$ such that
\begin{equation}
\label{qform1}
q(x\oplus y\oplus z)+q(z) -q(x\oplus z)-q(y\oplus z)=J(x,y)
\end{equation}
for all $z\in \calK$ and for all $x,y\in \calL(\calK)$.
\end{dfn}
Informally, Eq.~(\ref{qform1}) demands that a discrete analogue of the second derivative
$\frac{\partial^2 q}{\partial x \partial y}$ evaluated at some point $z\in \calK$ does not depend on $z$,
as it is the case for quadratic forms over the real field. 
The next  lemma states properties of the function $J(x,y)$
that follow from Eq.~(\ref{qform1}). 
\begin{lemma}
\label{lemma:qform1}
The function $J(x,y)$ defined by Eq.~(\ref{qform1}) is a symmetric bilinear form that 
takes values $0,4 \pmod 8$. Namely,  
$J(x,y)=J(y,x)$, $J(0,y)=0$,  
$J(x'\oplus x'',y)=J(x',y)+J(x'',y)$,
and $J(x,y)={0,4 \pmod 8}$
for all $x,x',x'',y\in \calL(\calK)$.
\end{lemma}
\begin{proof}
Let $x=x'\oplus x''$. 
Substituting $x\gets x'$
in Eq.~(\ref{qform1}) gives
\[
J(x',y)=q(x'\oplus y \oplus z) +q(z) - q(x'\oplus z) - q(y\oplus z).
\]
Substituting $x\gets x''$ and $z\gets z\oplus x'$ in Eq.~(\ref{qform1}) gives
\[
J(x'',y)=q(x\oplus y \oplus z) + q(x'\oplus z) - q(x\oplus z) - q(x'\oplus y \oplus z).
\]
This shows that $J(x,y)=J(x',y)+J(x'',y)$. 
The identities $J(0,y)=0$ and $J(x,y)=J(y,x)$ follow trivially from Eq.~(\ref{qform1}).
Replacing $z$ by $z\oplus x$ in  Eq.~(\ref{qform1}) yields
\[
J(x,y)=q(y\oplus z)+q(x\oplus z) - q(z)-q(x\oplus y \oplus z).
\]
Combining this and Eq.~(\ref{qform1}) one gets
$2J(x,y)=0$, that is,  $J(x,y)={0 \pmod 4}$. 
\end{proof}
As a corollary, one gets
$q(x\oplus z)-q(z)\in \{0,2,4,6\}$ for all $z\in \calK$ and for all $x\in \calL(\calK)$.
This can be checked by choosing $x=y$ in Eq.~(\ref{qform1})
and using the fact that $J(x,x)\in \{0,4\}$.

Suppose $g^1,\ldots,g^k\in \calL(\calK)$ is some fixed basis of $\calL(\calK)$,
$h\in \calK$ is some fixed shift vector, and $x\in \calK$. Then 
\[
x=h\oplus x_1g^1 \oplus \ldots \oplus x_k g^k, \quad x_i\in \{0,1\}.
\]
We shall write $\vec{x}\equiv (x_1,\ldots,x_k)$ to avoid
confusion between a point $x\in \calK$ and its 
coordinates.
Applying
Eq.~(\ref{qform1}) and Lemma~\ref{lemma:qform1} one can
describe $q$ in a basis-dependent way as 
\begin{equation}
\label{qform3}
q(\vec{x})=Q + \sum_{a=1}^k D_a x_a +  \sum_{1\le a<b\le k} J_{a,b} x_a x_b,
\end{equation}
where $Q\equiv q(h)\in \ZZ_8$, 
\begin{equation}
\label{qform4}
D_a=q(g^a\oplus h)-q(h) \in \{0,2,4,6\},
\end{equation}
\begin{equation}
\label{qform5}
J_{a,b}=J_{b,a} =J(g^a,g^b)\in \{0,4\}.
\end{equation}
We shall consider $J$ as a symmetric $k\times k$ matrix.
Although Eq.~(\ref{qform3}) depends only on off-diagonal matrix
elements of $J$, it will be convenient to retain the diagonal of $J$.
Combining Eqs.~(\ref{qform1},\ref{qform3}) one gets
\begin{equation}
\label{J(x,x)}
J_{a,a}=2D_a, \quad 1\le a\le k.
\end{equation}
 A connection between  quadratic forms and stabilizer states
is established by  the following lemma.
\begin{lemma}
\label{lemma:qform2}
Any $n$-qubit stabilizer state  can be uniquely written as 
\begin{equation}
\label{std}
 |\calK,q\rangle \equiv 2^{-k/2} \sum_{x\in \calK} e^{\frac{i\pi}4 q(x)}  |x\rangle,
\end{equation}
where $\calK\subseteq \FF_2^n$ 
is an affine space of dimension $0\le k\le n$
and $q\, : \, \calK \to \ZZ_8$ is a quadratic form.
\end{lemma}
\begin{proof}
The claim that any stabilizer state can be written 
in the form Eq.~(\ref{std}) 
follows from the explicit characterization of quadratic forms
Eqs.~(\ref{qform3},\ref{qform4},\ref{qform5}) 
and the canonical form of stabilizer states derived 
in Refs.~\cite{Nest2010,dehaene2003clifford,garcia2014hybrid}.
The uniqueness of the decomposition Eq.~(\ref{std}) is obvious.
\end{proof}
Next let us describe how the representation of $q$ 
transforms under various basis changes.
Suppose $R\in \FF_2^{k\times k}$ is an invertible  matrix.
Consider a basis change
\begin{equation}
\label{basis_change}
g^a\gets  \sum_{b=1}^k R_{a,b} \, g^b {\pmod 2}, 
\end{equation}
where $1\le a\le k$. The shift vector $h$ remains unchanged.
Applying Eq.~(\ref{qform3})
where $\vec{x}$ is chosen as the $a$-th row of $R$, 
one can easily check that the coefficients 
$(Q,D,J)$ 
transform according to $Q\gets Q$,
\begin{equation}
\label{updateJ1}
D_a \gets \sum_{b=1}^k R_{a,b} D_b + \sum_{1\le b<c\le k} J_{b,c} R_{a,b} R_{a,c},
\end{equation}
and
\begin{equation}
\label{updateJ2}
J\gets RJR^T.
\end{equation}
The matrix multiplications are performed in the ring $\ZZ_8$.
Next consider a basis change that  
alters the shift vector, 
\begin{equation}
\label{change_h}
h\gets h\oplus y, \quad \mbox{where} \quad
y=\sum_{a=1}^k y_a g^a {\pmod 2}.
\end{equation}
Using Eq.~(\ref{qform4})  one can easily check that the coefficients 
$(Q,D,J)$  transform according to
\begin{equation}
\label{updateJ4}
Q\gets Q+\sum_{a=1}^k D_a y_a + \sum_{1\le a<b\le k} J_{a,b} y_a y_b,
\end{equation}
\begin{equation}
\label{updateJ3}
D_a\gets D_a +  \sum_{b=1}^k J_{a,b} y_b,
\end{equation}
and $J\gets J$.

The above rules determine the
representation $(Q,D,J)$  of $q$ in any basis of $\calK$. 
What is the cost of computing this representation~?
Clearly, all updates can be expressed as a constant number of matrix-matrix (matrix-vector)
multiplications with $\ZZ_8$-valued matrices of size $k$.
Thus the updates have cost $O(k^3)$ in the worst case.
We shall often consider basis changes Eq.~(\ref{basis_change})
such that the matrix $R$ is sparse. Let $|R|$ be the total number
of non-zeros in $R$. Using sparse matrix-matrix multiplication
one can perform all updates in Eqs.~(\ref{updateJ1},\ref{updateJ2}) in time $O(|R|^2)$. 
Indeed, let $w_a$ be the number of non-zeros in the $a$-th row of $R$.
One can update $D_a$ and $J_{a,b}$ for any fixed $a,b$ in time
$w_a^2$ and $w_aw_b$ respectively. Thus $D$ and $J$ can be 
updated in time $O((\sum_{a=1}^k w_a)^2)=O(|R|^2)$.
Since the updates Eq.~(\ref{updateJ3},\ref{updateJ4}) require time $O(k^2)$
and $|R|\ge k$, the overall time is $O(|R|^2)$.
We conclude that computing  the representation $(Q,D,J)$ of $q$ in the new basis 
 takes time 
\begin{equation}
\label{update_cost}
\tau_{update}=O(\min{(k^3,|R|^2)}).
\end{equation}

In the rest of this section we show how to compute certain exponential
sums associated with quadratic forms, namely, 
\begin{equation}
\label{W}
W(q)\equiv \sum_{x\in \FF_2^k} e^{i\frac{\pi}4 q(\vec{x})},
\end{equation}
where $q(\vec{x})$ is defined by Eq.~(\ref{qform3}).
Of course, the addition in Eq.~(\ref{W}) is over the complex field. 
Our algorithm  takes as input the data $k,Q,D,J$ describing
$q(\vec{x})$ and outputs $W(q)$. The algorithm has running time $O(k^3)$.
It will be used  as a subroutine for computing the 
inner product between two stabilizer states, see Appendix~C.

It will be convenient to consider a more general sum
\begin{equation}
\label{W1}
W(\calK,q)=\sum_{x\in \calK} e^{i\frac{\pi}4 q(x)},
\end{equation}
where $\calK=\calL(\calK)\oplus h$ is an affine space
and $q\,: \, \calK\to \ZZ_8$ is a quadratic form on $\calK$. 
Clearly, Eq.~(\ref{W}) is a special case of Eq.~(\ref{W1}).
Let us say that $g^1,\ldots,g^k\in \calL(\calK)$ is a {\em  canonical basis} of $\calL(\calK)$
iff the set of basis vectors can be partitioned into disjoint
subsets
\begin{equation}
\label{subsets}
[k]=\calD_1\cup \ldots \cup \calD_r \cup M \cup S,
\end{equation}
such that
\begin{equation}
\label{c1}
|\calD_1|=\ldots=|\calD_r|=2, \quad |S|\le 1,
\end{equation}
\begin{equation}
\label{c0}
J_{a,a}=\left\{ \ba{rcl}
0 &\mbox{if} & a\notin S,\\
4 &\mbox{if} & a\in S.\\
\ea\right.
\end{equation}

\begin{equation}
\label{c2}
a\in M \quad \Rightarrow \quad J_{a,b}=0 \quad \forall b\in [k]\setminus S,
\end{equation}
and $\calD_i=\{a,b\}$ implies 
\begin{equation}
\label{c3}
J_{a,b}=4 \quad \mbox{and} \quad  J_{a,c}=J_{b,c}=0
\quad \forall c\notin S\cup \{a,b\}.
\end{equation}
Some of the subsets in Eq.~(\ref{subsets}) can be empty.
Let us write 
\[
\calD_j=\{ a(j), b(j)\}, \quad j=1,\ldots,r.
\]

Assume that $\calL(\calK)$ is already equipped with a canonical basis
$g^1,\ldots,g^k$
and show how to compute the sum $W(\calK,q)$.
Suppose first $S=\emptyset$. 
By repeatedly applying Eq.~(\ref{qform3}) 
and using Eqs.~(\ref{c1},\ref{c2},\ref{c3})
one can check that 
\begin{equation}
\label{canonical1}
q(\vec{x})=Q+  \sum_{j=1}^r q_j(x_{a(j)},x_{b(j)})
+ \sum_{c\in M} D_c x_c.
\end{equation}
where $q_j\, : \, \FF_2^2 \to \ZZ_8$ is defined by
\begin{equation}
\label{canonical2}
q_j(y,z)=4yz + D_{a(j)} y + D_{b(j)} z.
\end{equation}
Examination of Eqs.~(\ref{W1},\ref{canonical1},\ref{canonical2}) reveals that
the sum $W(\calK,q)$ factorizes into a product of $O(k)$ terms such that each 
term can be computed in time $O(1)$. 
Specifically,
\begin{equation}
\label{canonical3}
W(\calK,q)=e^{i\frac{\pi}4 Q} \cdot \prod_{c\in M} \left(
1+e^{i\frac{\pi}4 D_{c} } \right)\cdot \prod_{j=1}^r \Gamma_j,
\end{equation}
where
\begin{equation}
\label{canonical4}
\Gamma_j= 1+ e^{i\frac{\pi}4 D_{a(j)}} + e^{i\frac{\pi}4 D_{b(j)}}- e^{i\frac{\pi}4( D_{a(j)}+D_{b(j)})}
\end{equation}
Combining Eqs.~(\ref{canonical3},\ref{canonical4})
one can compute 
$W(\calK,q)$ in time $O(k)$. 
 
Consider now the remaining case  $S\ne \emptyset$.
Since $|S|\le 1$, we have $S=\{s\}$ for some  $s\in [k]$.
By repeatedly applying Eq.~(\ref{qform3}) 
and using Eqs.~(\ref{c1},\ref{c2},\ref{c3})
one can check that 
\begin{eqnarray}
\label{canonical5}
q(\vec{x})&=&Q+ D_s x_s+ \sum_{j=1}^r q_j(x_{a(j)},x_{b(j)},x_s) \nonumber \\
&&+ \sum_{c\in M } (D_c x_c + J_{c,s}x_c x_s),
\end{eqnarray}
where $q_j\, : \, \FF_2^3 \to \ZZ_8$ is defined by
\begin{equation}
\label{canonical6}
q_j(y,z,\sigma)=4yz + J_{a(j),s} y\sigma + J_{b(j),s} z\sigma + D_{a(j)} y + D_{b(j)} z.
\end{equation}
We have $W(\calK,q)=W_0+W_1$, where 
\begin{equation}
\label{split}
W_\sigma\equiv \sum_{x\in \calK\, : \, x_s=\sigma}\; e^{i\frac{\pi}4 q(\vec{x})}, \quad \sigma=0,1.
\end{equation}
Examination of Eqs.~(\ref{canonical5},\ref{canonical6}) reveals that
$W_\sigma$ factorizes into a product of $O(k)$ terms such that each 
term can be computed in time $O(1)$. 
Specifically,
\begin{equation}
\label{canonical7}
W_\sigma=e^{i\frac{\pi}4\left( Q + \sigma D_s\right)}\prod_{c\in M} \left(
1+e^{i\frac{\pi}4 (D_c+\sigma J_{c,s}) } \right)\cdot \prod_{j=1}^r \Gamma_j(\sigma),
\end{equation}
where 
\begin{eqnarray}
\label{canonical8}
\Gamma_j(\sigma)&=& 1+ \exp{\left[ i\frac{\pi}4\left( J_{a(j),s} \sigma  + D_{a(j)} \right)\right]}   \\
&&+ \exp{\left[i\frac{\pi}4\left( J_{b(j),s} \sigma  + D_{b(j)} \right)\right]} \nonumber \\
&& - \exp{\left[ i\frac{\pi}4\left( J_{a(j),s} \sigma  + J_{b(j),s} \sigma  + D_{a(j)} +D_{b(j)}  \right)\right]}.\nonumber
\end{eqnarray}
Combining Eqs.~(\ref{canonical7},\ref{canonical8})
one can compute $W_0+W_1$ in time $O(k)$. 

To transform an arbitrary basis
$g^1,\ldots,g^k$  of $\calL(\calK)$ into the canonical form
we shall use a version of the Gram-Schmidt orthogonalization.
It  involves at most $k$ basis changes 
Eq.~(\ref{basis_change}) with sparse matrices $R$ such that $|R|=O(k)$.
Computing the coefficients $(D,J)$ in the canonical basis thus takes time $O(k|R|^2)=O(k^3)$,
see Eq.~(\ref{update_cost}).

Recall that $D_a\in \{0,2,4,6\}$.
Define a subset 
\[
S=\{a\in [k]\, : \, D_a\in \{2,6\} \}.
\]
If $S$ is non-empty, pick an arbitrary element $s\in S$.
Perform a basis change $g^a\gets g^a\oplus g^s$ for each $a\in S\backslash s$.
From Eq.~(\ref{updateJ1})  one gets
$D_a\gets D_a + D_s + J_{a,s}\in \{0,4\}$ for all $a\in S\backslash s$
and $D_a\gets D_a$ for all $a\notin S$.
Set $S=\{s\}$.
Now we can assume that 
$D_a\in \{0,4\}$ for all $a\notin S$ for some subset $S$ such that
$|S|\le 1$. From Eq.~(\ref{J(x,x)}) we infer 
\begin{equation}
\label{canonical9}
J_{a,a}=0  \quad \mbox{for all $a\notin S$}.
\end{equation}
Let us say that a pair of basis vectors $(g^a,g^b)$ with $a,b\notin S$
is a dimer if  it obeys Eq.~(\ref{c3}), that is,
$J_{a,b}=4$ and $J_{a,c}=J_{b,c}=0$
for all $c\notin S\cup \{a,b\}$.
Note that a basis vector can belong to at most one dimer.
Let us say that a basis vector $g^a$ 
 with $a\notin S$ is 
a monomer if it obeys Eq.~(\ref{c2}), that is,
 $J_{a,b}=0$ for all $b\in [k]\setminus S$. 
Partition the set of basis vectors into four disjoint
sets, 
\begin{equation}
\label{canonical11}
[k]=\calD \cup M \cup S \cup E,
\end{equation}
such that $\calD$ is the union of all dimers,
$M$ is the union of all monomers, 
and $E$ 
is the complement of $\calD MS$.
By definition, a basis has a canonical form iff $E$ is empty.
Initially $\calD$, $M$ are empty, and $E$
is the complement of $S$.
Suppose $E$ is non-empty.
Pick any $a\in E$.
If $J_{a,b}=0$ for all $b\in E$, move $a$
from $E$ to $M$.
Otherwise $J_{a,b}=4$ for some $b\in E$.
Let us define a binary matrix $\pmb{J}$ corresponding to $J$
such that $\pmb{J}_{a,b}=1$ if $J_{a,b}=4$ and
$\pmb{J}_{a,b}=0$ otherwise.
Perform a basis change
\begin{equation}
\label{canonical12}
g^c\gets g^c \oplus \pmb{J}_{a,c} g^b \oplus \pmb{J}_{b,c} g^a \quad \mbox{for all $c\in E \setminus \{a,b\}$}.
\end{equation}
Using Eq.~(\ref{canonical9}) one can check that 
the new basis vectors obey $J(g^c,g^a)=J(g^c,g^b)=0$
for all $c\in \calD ME\setminus \{a,b\}$.
Thus we can
move $a,b$ from $E$ to $\calD$
by creating a new dimer $\calD_i=\{a,b\}$ in Eq.~(\ref{subsets}). 
By repeating the above steps  at most $k$ times one makes $E=\emptyset$.
Furthermore, the $R$ matrices corresponding to the basis change
Eq.~(\ref{canonical12}) are sparse since any row of $R$
contains at most three non-zero elements.  
Thus the original basis is transformed into the canonical form
by $O(k)$ basis changes Eq.~(\ref{basis_change}) with sparse matrices $R$
such that $|R|=O(k)$. This has cost $O(k|R|^2)=O(k^3)$.
We summarize the algorithm below.

\begin{center}
\fbox{\parbox{0.9\linewidth}{
\begin{algorithmic}
\Function{ExponentialSum}{$Q,D,J$} 
\State{$S\gets \{ a\in [k]\, : \, D_{a}\in \{2,6\}\}$}
\If{$S\ne \emptyset$}
  \State{Pick any $a\in S$}
  \For{$b\in S\setminus \{a\}$}
  \State{$g^b\gets g^b\oplus g^a$}
  \EndFor
  \State{Update $(D,J)$ using Eqs.~(\ref{updateJ1},\ref{updateJ2})}
  \State{$S\gets \{a\}$}
\EndIf
\State{\Comment{Now $J_{a,a}=0$ for all $a\notin S$}}
\State{$E\gets [k]\setminus S$}
\State{$M\gets \emptyset$}
\State{$r\gets 0$}
\While{$E\ne \emptyset$}
 \State{Pick any $a\in E$}
 \State{$K\gets \{b\in E\setminus a\, : \, J_{a,b}=4\}$}
 \If{$K=\emptyset$}
 \State{\Comment{Found a new monomer $\{a\}$}}
   \State{$M\gets M\cup a$}
   \State{$E\gets E\setminus a$}
  \Else
    \State{Pick any $b\in K$}
    \For{$c\in E\setminus \{a,b\}$}
    \State{$g^c \gets g^c \oplus \pmb{J}_{a,c} g^b \oplus \pmb{J}_{b,c}g^a$}
    \EndFor
     \State{Update $(D,J)$ using Eqs.~(\ref{updateJ1},\ref{updateJ2})}
     \State{\Comment{Now $\{a,b\}$ form a new dimer}}
       \State{$r\gets r+1$, $\calD_r\gets \{a,b\}$}
     \State{$E\gets E\setminus \{a,b\}$}
 \EndIf
\EndWhile
\If{$S=\emptyset$}
 \State{Compute $W(\calK,q)$ from Eq.~(\ref{canonical3})}
\Else
 \State{Compute $W_{0,1}$ from Eq.~(\ref{canonical7})}
 \State{Set $W(\calK,q)=W_0+W_1$}
\EndIf
 \EndFunction
\end{algorithmic}
}}
\end{center}
{\em Comments:}  The basis vectors $g^a$ only serve a 
notational purpose to  describe the basis change matrix $R$
that must be used in the update formulas Eqs.~(\ref{updateJ1},\ref{updateJ2}).
There are no actual data representing $g^a$ or operations performed with them.
For example, the first {\bf for} loop corresponds to a matrix
$R=I\oplus \sum_{b\in S\setminus \{a\}} (e^b)^T e^a$, where $e^a$
is the binary vector with a single `$1$' at the $a$-th position.
As was shown in Ref.~\cite{BSS15}, the sum
$W(\calK,q)$ can be represented by a triple of integers
$p\ge 0$, $m\in \ZZ_8$, and $\epsilon\in \{0,1\}$ such that
$W(\calK,q)=\epsilon \cdot 2^{p/2}\cdot e^{i\pi m/4}$.
 Our implementation of the algorithm uses such representation
for all intermediate sums to avoid roundoff errors.
Timing analysis  for a MATLAB implementation  is reported in Table~1.

\section*{Appendix B: Standard form of stabilizer states}
\label{app:B}

Suppose 
$|\calK,q\rangle\in \calS_n$ is a stabilizer state of $n$ qubits defined in Eq.~(\ref{std}).
An affine space $\calK=\calL(\calK)\oplus h \subseteq\FF_2^n$ 
of dimension $k$ 
will be represented by a tuple 
\[
(n,k,h\in \FF_2^n,G,\bar{G}\in \FF_2^{n\times n}),
\]
such that $\calL(\calK)$ is spanned by the first $k$ rows
of the matrix $G$ and $\bar{G}\equiv (G^{-1})^T$, that is,
\begin{equation}
\label{GGbar}
G\bar{G}^T=I {\pmod 2}.
\end{equation}
We shall write  $g^a$ and $\bar{g}^a$
for the $a$-th row of $G$ and $\bar{G}$ respectively.
Thus $\calL(\calK)=\mathrm{span}(g^1,\ldots,g^k)$
and $(g^a,\bar{g}^b)=\delta_{a,b}$ for $1\le a,b\le n$.
We shall refer to $g^a$ and $\bar{g}^a$ as the primal and 
the dual basis vectors. 

A quadratic form $q\, : \, \calK\to \ZZ_8$ will be specified
by a list of coefficients  $(Q,D,J)$
that describe  $q(\vec{x})$ in the 
basis $g^1,\ldots,g^k$ of $\calL(\calK)$, see 
Eqs.~(\ref{qform3},\ref{qform4},\ref{qform5}),
with the shift vector $h$.
Thus, a stabilizer state $|\calK,q\rangle$ of $n$ qubits
is described by the following data:
\[
(n,k,h,G,\bar{G},Q,D,J),
\]
where $Q\in \ZZ_8$, $D_1,\ldots,D_k\in \{0,2,4,6\}$,
and $J$ is a symmetric $k\times k$
such that $J_{a,b}\in \{0,4\}$ for all $a,b$.
A valid data must satisfy conditions Eq.~(\ref{GGbar})
and Eq.~(\ref{J(x,x)}).

We shall often use a subroutine that 
alters a stabilizer state $|\calK,q\rangle$ by shrinking 
the affine space $\calK$ reducing its dimension by one. Namely, 
consider a vector $\xi\in \FF_2^n$ and $\alpha\in \FF_2$.
Define 
\begin{equation}
\label{shrink}
\calM=\calK \cap \{ x\in \FF_2^n \, : \, (\xi,x)=\alpha \}.
\end{equation}
Clearly, $\calM$ is an affine space which is either empty,
or $\calM=\calK$, or $\calM$ has dimension $k-1$.  
Below we describe an algorithm that takes as input a stabilizer
state $|\calK,q\rangle$ and computes the standard form of
the state  $|\calM,q\rangle$ 
(or reports that $\calM$ is empty).
Here it is understood that the form $q$ is restricted onto $\calM$.
The algorithm has runtime $O(kn)$.
First we note that 
\[
\calM=h\oplus \{ y\in \calL(\calK)\, : \, (\xi,y)=\beta \},
\]
where $\beta=\alpha\oplus (\xi,h)$.
Let  
\[
S=\{ a\in [k]\, : \,  (\xi,g^a)=1\}.
\]
One can compute $S$ in time $O(kn)$. 
If $S=\emptyset$ and $\beta=1$ then $\calM$ is empty.
If $S=\emptyset$ and $\beta=0$ then $\calM=\calK$.
Otherwise pick any element $i\in S$
and remove $i$ from $S$.
Change the basis of $\calL(\calK)$ according to 
\[
g^a\gets g^a \oplus g^i \quad \mbox{for $a\in S$}.
\]
Change the dual basis according to
\[
\bar{g}^i \gets \bar{g}^i \oplus \sum_{a\in S} \bar{g}^a.
\]
Now $(g^a,\bar{g}^b)=\delta_{a,b}$ for all $a,b$.
The basis change requires time $O(kn)$. 
Let us also swap the $i$-th and the $k$-th basis vectors.
Updating the coefficients $(D,J)$  using Eqs.~(\ref{updateJ1},\ref{updateJ2})
takes time $O(k^2)=O(kn)$. 
Now  basis vectors $g^1,\ldots,g^{k-1}$  are orthogonal to $\xi$
and $(\xi,g^k)=1$. Thus 
\[
\calM=h' \oplus \mathrm{span}(g^1,\ldots,g^{k-1})\equiv h'\oplus \calL(\calM),
\]
where $h'=h\oplus \beta g^k$
is the new shift vector.  
Update the coefficients $(Q,D)$ 
using Eqs.~(\ref{updateJ4},\ref{updateJ3}), where 
$y=\beta g^k$. This takes time $O(k)$.
Now restricting the form $q$ onto $\calM$ is equivalent to removing the $k$-th row/column from the matrix $J$
and removing the $k$-th element from $D$.
We obtained the standard form of the state $|\calM,q\rangle$.
The entire  algorithm is summarized below.

\begin{center}
\fbox{\parbox{0.9\linewidth}{
\begin{algorithmic}
\Function{Shrink}{$|\calK,q\rangle,\xi,\alpha$}
\State{$S\gets \{ a\in [k] \, : \,  (\xi,g^a)=1\}$}
\State{$\beta \gets \alpha\oplus (\xi,h)$}
\If{$S=\emptyset$ and $\beta=1$}
\State{\Return{EMPTY}}
\EndIf
\If{$S=\emptyset$ and $\beta=0$}
\State{\Return{SAME}}
\EndIf
\State{Pick any $i\in S$}
\State{$S\gets S\setminus \{i\}$}
\For{$a\in S$}
\State{ $g^a\gets g^a \oplus g^i$}
\State{Update $(D,J)$ using Eqs.~(\ref{updateJ1},\ref{updateJ2})}
\EndFor
\State{$\bar{g}^i \gets  \bar{g}^i \oplus \sum_{a\in S} \bar{g}^a$}
\State{Swap $g^i$ and $g^k$. Swap $\bar{g}^i$ and $\bar{g}^k$.}
\State{Update $(D,J)$ using Eqs.~(\ref{updateJ1},\ref{updateJ2})}
\State{$h\gets h\oplus \beta g^k$}
\State{Update $(Q,D)$ using Eqs.~(\ref{updateJ4},\ref{updateJ3})}
\State{Remove the $k$-th row/column from $J$}
\State{Remove the $k$-th element from $D$}
\State{$k\gets k-1$}
\State{\Return{SUCCESS}}
\EndFunction
\end{algorithmic}
}}
\end{center}
To simplify notations, here we assume that the function SHRINK
modifies the data describing the input state.
The function reports whether the new affine space 
$\calK$ is empty or the same as the initial space.
It reports SUCESS whenever the dimension of the affine space
has been reduced by one. 
The function has runtime $O(kn)$.
Sometimes we shall use a ``lazy" version of the function that
does not update the coefficients of $q$. We shall use the notation
SHRINK${}^*$ for such lazy version.

 \section*{Appendix C: The inner product}
\label{app:C}

Consider a pair of $n$-qubit stabilizer states 
\[
|\phi_\alpha \rangle=|\calK_\alpha, q_\alpha\rangle, \quad \alpha=1,2
\]
with the standard forms 
$(n,k_\alpha,h_\alpha,G_\alpha,\bar{G}_\alpha,Q_\alpha,D_\alpha, J_\alpha)$.
Below we describe an algorithm that computes the inner product
\begin{equation}
\label{inner1}
\langle \phi_2 |\phi_1\rangle = 2^{-(k_1+k_2)/2} \sum_{x\in \calK_1\cap \calK_2} e^{i\frac{\pi}4 (q_1(x)-q_2(x))}.
\end{equation}
in time $O(n^3)$.
First we note that $x\in \calK_2$ iff 
\[
x\oplus h_2\in \calL(\calK_2)=\mathrm{span}(g_2^1,\ldots,g_2^{k_2}).
\]
Thus $x\in \calK_2$ iff $x\oplus h_2$ 
is orthogonal to all
dual basis vectors $\bar{g}_2^a$ with $k_2<a\le n$.
Here and below $g_\alpha^b$ and $\bar{g}_\alpha^b$ denote
the $b$-th row of $G_\alpha$ and $\bar{G}_\alpha$ respectively.
Thus 
\[
\calK\equiv \calK_1\cap \calK_2=\bigcap_{b=k_2+1}^n \{ x\in \calK_1 \, : \,  (\bar{g}_2^b,x)=(h_2,\bar{g}_2^b) \}.
\]
One can compute the standard form of $|\calK,q_1\rangle$
by $n-k_2$ calls to the function SHRINK defined in Appendix~B with
$\xi=\bar{g}_2^b$ and $\alpha=(h_2,\bar{g}_2^b)$ 
for $b=k_2+1,\ldots,n$. This takes time  
\[
\tau_1=O((n-k_2)k_1n)
\]
since we have to call SHRINK  $n-k_2$ times.

Let $\calK=(n,k,h,G,\bar{G})$ be the standard form of $\calK$ and $(Q_1,D_1,J_1)$  be the coefficients of $q_1$
restricted onto $\calK$ in the  basis $g^1,\ldots,g^k$
(as usual, $g^a$ is the $a$-th row of $G$).  

The next step so the compute coefficients $(Q_2,D_2,J_2)$  of the form  $q_2$ restricted to $\calK$
in the basis $g^1,\ldots,g^k$ with the shift vector $h$.
We note that 
\[
h=h_2\oplus \sum_{a=1}^{k_2} y_a g_2^a, \quad \mbox{where} \quad y_a=(h\oplus h_2,\bar{g}_2^a).
\]
One can compute $y_1,\ldots,y_{k_2}$   in time $O(k_2 n)$ and then 
compute the updated coefficients  $(Q_2,D_2)$ 
from Eqs.~(\ref{updateJ4},\ref{updateJ3}).
This takes time $O(k_2^2)$.
A simple algebra shows that $\calL(\calK)=\calL(\calK_1)\cap \calL(\calK_2)$,
that is, $g^a\in \calL(\calK_2)$ for all $a=1,\ldots,k$.
Define a matrix $R$ of size $k\times k_2$ such that 
\[
g^a=\sum_{b=1}^{k_2} R_{a,b} g_2^b {\pmod 2}, \quad 1\le a\le k.
\]
Using the dual basis of $\calK_2$ one gets
$R_{a,b}=(g^a,\bar{g}_2^b)$.
One can compute the entire matrix $R$ in time $O(kk_2n)$.
Then the  coefficients $(D_2,J_2)$  in the basis $g^1,\ldots,g^k$
can be computed from Eqs.~(\ref{updateJ1},\ref{updateJ2})
which takes time
$O(kk_2^2)$, see Eq.~(\ref{update_cost}).
(Here we used a slightly stronger version of Eq.~(\ref{update_cost})
taking into account that $R$ is a rectangular matrix.)
The runtime up to this point is 
\[
\tau_2=\tau_1+   O(kk_2n).
\]
Now the restrictions of both forms $q_1,q_2$ onto $\calK$
are defined in the same basis $g^1,\ldots,g^k$ and the same shift vector
$h$. Thus $q\equiv q_1-q_2$ has coefficients
$(Q,D,J)$, where $Q=Q_1-Q_2$, $D=D_1-D_2$, and $J=J_1-J_2$.
We get
\[
\langle \phi_2|\phi_1\rangle = 2^{-(k_1+k_2)/2} \cdot W(Q,D,J),
\]
where $W(Q,D,J)$ is the exponential sum Eq.~(\ref{W})
that can be computed in time $O(k^3)$, see Appendix~A.
The overall running time is thus
\[
\tau=\tau_2+O(k^3)=O( (n-k_2)k_1n + kk_2n  +k^3)=O(n^3).
\]
We summarize the entire inner product algorithm below. 

\begin{center}
\fbox{\parbox{0.9\linewidth}{
\begin{algorithmic}
\Function{InnerProduct}{$(|\calK_1,q_1\rangle, |\calK_2,q_2\rangle$} 
\State{$\calK\gets \calK_1$}
\For{$b=k_2+1$ to $n$}
\State{$\alpha\gets (h_2,\bar{g}_2^b)$}
\State{$\epsilon \gets$SHRINK$(|\calK,q_1\rangle,\bar{g}_2^b,\alpha)$}
\If{$\epsilon=$EMPTY}
\State{\Return{$0$}}
\EndIf
\EndFor
\State{\Comment{Now $\calK=\calK_1\cap \calK_2=(n,k,h,G,\bar{G})$}}
\For{$a=1$ to $k_2$}
\State{$y_a\gets (h\oplus h_2,\bar{g}_2^a)$}
\For{$b=1$ to $k$}
\State{$R_{b,a}\gets (g^b,\bar{g}_2^a)$}
\EndFor
\EndFor
\State{$h_2\gets h_2\oplus \sum_{a=1}^{k_2} y_a g_2^a=h$}
\State{Update $(Q_2,D_2)$ using Eqs.~(\ref{updateJ4},\ref{updateJ3}) with $y$}
\State{Update $(D_2,J_2)$ using Eqs.~(\ref{updateJ1},\ref{updateJ2}) with $R$}
\State{\Comment{Now $q_1$, $q_2$ are defined in the same basis}}
\State{$Q\gets Q_1-Q_2$}
\State{$D\gets D_1-D_2$}
\State{$J\gets J_1-J_2$}
\State{\Return{$2^{-(k_1+k_2)/2} \cdot$ExponentialSum$(Q,D,J)$}}
\EndFunction
\end{algorithmic}
}}
\end{center}
{\em Comments:} As before, we assume that 
the output is converted to a triple of integers $(\epsilon,p,m)$ such that 
$\langle \phi_2|\phi_1\rangle=\epsilon \cdot 2^{p/2}\cdot e^{i\pi m/4}$.
 If both $k_1$ and $k_2$ are small, one can compute the intersection
$\calK_1\cap \calK_2$ directly by solving a linear system 
\[
\sum_{a=1}^{k_1} x_a g_1^a \oplus \sum_{b=1}^{k_2} y_b g_2^b = h_1 \oplus h_2
\]
with $k_1+k_2$ variables and $n$ equations. This 
provides a shift vector and a basis for $\calK$ in time $O(n(k_1+k_2)^2)$.
Then one can compute the updated coefficients of $q_1$ and $q_2$
in the new basis in time $O(k(k_1^2+k_2^2))$.
Thus the overall running time is 
\begin{equation}
\label{fast2}
\tau=O(k_1^2n +k_2^2 n  + k^3)
\end{equation}
which is linear in $n$ provided that both $k_1,k_2=O(1)$.
We note however that the vast majority of stabilizer states have
$k_\alpha\approx n$, see Appendix~D, so the above method provides no speedup
in the generic case. 

The timing analysis of the function InnerProduct reported in Table~1 was performed for 
inner products $\langle\tilde{x}|\phi\rangle$,
where $\phi\in \calS_n$ is drawn from the uniform distribution (as described in Appendix~D),
$x\in \FF_2^n$ is a random uniformly distributed string,
and  $|\tilde{x}\rangle\equiv |\tilde{x}_1\otimes \tilde{x}_2\otimes \cdots \otimes \tilde{x}_n\rangle$,
where $|\tilde{0}\rangle=|0\rangle$ and $|\tilde{1}\rangle=H|0\rangle$.
This choice is justified since our simulation algorithm only requires 
inner products of the above form.  

\section*{Appendix D: Random stabilizer states}
\label{app:D}

Let us now describe
an algorithm that generates 
a random uniformly distributed stabilizer state $|\calK,q\rangle \in \calS_n$.
The algorithm has average-case runtime $O(n^2)$
and the worst-case runtime $O(n^3)$.

For each $0\le k\le n$ 
define a subset of stabilizer states
\[
\calS_n^k = \{ |\calK,q\rangle \in \calS_n \, : \, \dim{(\calK)}=k \}.
\]
For example, $\calS_n^0$ includes
all basis vectors, whereas 
$\calS_n^n$ includes stabilizer states
supported on all basis vectors.
Our algorithm first picks a random integer $d=0,1,\ldots,n$ 
drawn from a distribution 
\begin{equation}
\label{P1}
P(d)=\frac{|\calS_n^{n-d}|}{\sum_{m=0}^n |\calS_n^m|}
\end{equation}
and generates a random subspace $\calK\subseteq \FF_2^n$
of dimension $k=n-d$.
To compute $P(d)$ we need the following fact.
\begin{lemma}
\label{lemma:count}
\begin{equation}
\label{count}
|\calS_n^{n-d}|=8\cdot 2^{n+\frac12 \left[ n(n+1)-d(d+1)\right] } \cdot \prod_{a=1}^d \frac{1-2^{d-n-a}}{1-2^{-a}}.
\end{equation}
for any $d=1,\ldots,n$ and $|\calS_n^n|=8\cdot 2^{n+\frac12 n(n+1)}$.
\end{lemma}
\begin{proof}
Let $k\equiv n-d$.
The number of $k$-dimensional linear subspaces $\calL\subseteq \FF_2^n$ is known to be
\[
\Gamma_n^k = \Gamma_n^d=\prod_{m=0}^{d-1} \frac{2^n-2^m}{2^d-2^m}
\]
For a given $\calL$ there are $2^{n-k}$ affine spaces $\calK$
such that $\calK=\calL\oplus h$ for some shift vector $h$.
Finally, for a given affine space $\calK$ there are
\[
\Lambda_n^k = 8 \cdot 2^{2k} \cdot 2^{k(k-1)/2}
\]
quadratic forms $q\, : \, \calK\to \ZZ_8$.
Here the three factors represent the number of choices for
the coefficients $(Q,D,J)$ in Eqs.~(\ref{qform3},\ref{qform4},\ref{qform5})
respectively (recall that the diagonal of $J$ is determined by $D$,
see Eq.~(\ref{J(x,x)})).
It follows that $|\calS_n^k|=2^{n-k} \cdot \Gamma_n^k \cdot \Lambda_n^k$,
which gives Eq.~(\ref{count}).
\end{proof}
One can rewrite Eq.~(\ref{P1}) as 
\begin{equation}
\label{P2}
P(d)= \frac{\eta(d)}{\sum_{m=0}^n \eta(m)}, 
\end{equation}
where $\eta(0)=1$ and 
\[
\eta(d)=2^{-d(d+1)/2} \cdot \prod_{a=1}^d \frac{1-2^{d-n-a}}{1-2^{-a}}
\]
for $d=1,\ldots,n$.
One can compute a lookup table for the function $\eta(d)$
offline since it depends only on $n$. 
Clearly, $d=O(1)$ with high probability.
Thus, the average-case online complexity of sampling $d$  from 
the distribution $P(d)$ is $O(1)$. 

We start by choosing the zero shift vector such that $\calK$
is a random linear space of dimension $k$. 
We shall generate $\calK$ by repeatedly picking 
a random matrix $X\in \FF_2^{d\times n}$ until
$X$ has rank $d$ and then choosing 
$\calK=\ker{(X)}$. 
It is well-known that $X$ has rank $d$ with probability
\[
p_{n,d}=\prod_{a=0}^{d-1} (1-2^{-n+a}) \ge \max{ \{ 1/4, 1-2^{-n+d}\}}.
\]
Note that $p_{n,d}$ is exponentially close to $1$ whenever $d=O(1)$.
Thus $X$ has full rank after $O(1)$ attempts with high probability. 
Furthermore, one can compute
the rank of $X$ in time $O(nd^2)$ using the Gaussian elimination
by bringing $X$ into the row echelon form.
It is also well-known that conditioned on $X$ having full rank,
the subspace $\ker{(X)}$ is distributed uniformly on the set of
all  subspaces of $\FF_2^n$ of dimension $n-d$.
Thus we can choose $\calK=\ker{(X)}$.

The next step is computing $n\times n$ matrices $G$ and $\bar{G}$
such that $\calK$ is spanned by the first $k$ rows of $G$
and $G\bar{G}^T=I$.  Let us first set $\calK=\FF_2^n$ and $G=\bar{G}=I$.
Choose a zero quadratic form $q(x)=0$ for all $x\in \calK$.
Let $\xi^a$ be the $a$-th row of the matrix $X$.
One can make $\calK$ orthogonal to $\xi^1,\ldots,\xi^d$
by making $d$ calls to the  function
SHRINK${}^*(|\calK,q\rangle,\xi^a,0)$ defined in Appendix~B. 
(Recall that SHRINK${}^*$ does not update the coefficients of $q$.)
Finally we shift $\calK$ by a random uniformly distributed vector $h\in \FF_2^n$.
At this point $\calK$ is a random affine space represented in the standard form.
It remains to choose random coefficients of the  quadratic
form $q\, : \, \calK\to \ZZ_8$ in the basis $g^1,\ldots,g^k$.
Since $q$ must be distributed uniformly on the set of 
all quadratic forms $q\, : \, \calK\to \ZZ_8$,
we must choose $Q\in \ZZ_8$, $D_a\in \{0,2,4,6\}$,
and $J_{a,b}\in \{0,4\}$ for $a<b$ as random uniform elements of the respective sets.
Then the entire matrix $J$ is determined by $J_{b,a}=J_{a,b}$
and $J_{a,a}=2D_{a,a}$, see Eq.~(\ref{J(x,x)}).
The  entire algorithm is summarized below.

\begin{center}
\fbox{\parbox{0.9\linewidth}{
\begin{algorithmic}
\Function{RandomStabilizerState}{$n$} 
\State{Compute $P(0),\ldots,P(n)$ from Eq.~(\ref{P2})}
\State{Sample $d\in \{0,1,\ldots,n\}$  from $P(d)$}
\State{$k\gets n-d$}
\Repeat
\State{Pick random $X\in \FF_2^{d\times n}$}
\Until{$\mathrm{rank}{(X)}=d$}
\State{$G\gets I$, $\bar{G}\gets I$, $h\gets 0^k$}
\State{$\calK\gets (n,k,h,G,\bar{G})$}
\State{\Comment{Now $\calK=\FF_2^n$ is full binary space}}
\State{$q\gets$ all-zeros function on $\calK$}
\For{ $a=1$  to $d$}
\State{$\xi\gets$ $a$-th row of $X$}
\State{SHRINK${}^*(|\calK,q\rangle,\xi,0)$}
\EndFor
\State{\Comment{Now $\calK=\ker{(X)}$}}
\State{\Comment{$\calK$ has the standard form}}
\State{Pick random  $h\in \FF_2^n$}
\State{Pick random  $Q\in \ZZ_8$}
\State{Pick random $D_a\in \{0,2,4,6\}$}
\State{Pick random $J_{a,b}=J_{b,a}\in \{0,4\}$ for $a\ne b$}
\State{Set $J_{a,a}=2D_a {\pmod 8}$}
\State{\Return{$(n,k,h,G,\bar{G},Q,D,J)$}}
 \EndFunction
\end{algorithmic}
}}
\end{center}
Each call to SHRINK takes
time $O(n^2)$, see Appendix~B, whereas 
each computation of $\mathrm{rank}(X)$ takes time $O(dn^2)$.
Thus  the entire algorithm takes time $O(dn^2)$.
Since
$d=O(1)$ with high probability, see above, 
the average runtime is $O(n^2)$,
whereas the worst-case runtime is $O(n^3)$.
Timing analysis for a MATLAB implementation is reported in Table~1.

\section*{Appendix E: Pauli measurements}
\label{app:E}

Suppose $|\calK,q\rangle\in \calS_n$ is a stabilizer state
of $n$ qubits 
represented in the standard form and $P\in \calP_n$ is a  Pauli operator.
Define  an operator
\[
P_+ \equiv \frac12 (I+ P).
\]
It is well-known that $P_+$ maps stabilizer states to (unnormalized)
stabilizer states. 
Note that $P_+$ is a projector if $P$ is self-adjoint and
$\sqrt{2}P_+$ is a unitary Clifford operator if $P^\dag=-P$. 
Below we describe an algorithm that computes
the normalization and the standard form of the 
state $P_+|\calK,q\rangle$.
The algorithm has runtime $O(n^2)$. 
We shall be mostly interested in the case when $P_+$ is a projector
(although our algorithm applies to the general case). 
Note that a projector onto the codespace of any stabilizer code
with a stabilizer group
$\calG\subseteq \calP_n$ can be written as a product of at most $n$
projectors $P_+$ associated with some set of generators of $\calG$.
Thus a projected state $\Pi_\calG |\calK,q\rangle$ can be computed in time $O(n^3)$
using the above algorithm. 

Let $\calK=(n,k,h,G,\bar{G})$ be the standard form of $\calK$ and
\begin{equation}
\label{proj1}
P=i^m Z(\zeta)X(\xi), \quad m\in \ZZ_4, \quad  \xi,\zeta\in \FF_2^n.
\end{equation}
We shall consider two cases
depending on whether or not $\xi\in \calL(\calK)$. This inclusion can be checked
in time $O(kn)$ by computing inner products 
$\xi_a=(\xi,\bar{g}^a)$ with $a=1,\ldots,k$. Namely,
$\xi \in \calL(\calK)$ iff 
$\xi=\sum_{a=1}^k \xi_a g^a {\pmod 2}$.

\noindent
{\em Case~1:} $\xi \in \calL(\calK)$. 
Define a function 
\begin{equation}
\label{proj6}
\chi(x)=q(x\oplus \xi) -q(x).
\end{equation}
By definition of a quadratic form one has 
\begin{equation}
\label{proj7}
\chi(h\oplus y) = \chi(h) + J(\xi,y) \quad \mbox{for all $y\in \calL(\calK)$}.
\end{equation}
The state $P_+|\calK,q\rangle$ can be written  as
\begin{equation}
\label{proj8}
2^{-k/2-1}\sum_{x\in \calK} e^{i\frac{\pi}4 q(x)} \left(
1+ i^m (-1)^{(\zeta, x)} e^{i\frac{\pi}4 \chi(x)}\right)|x\rangle.
\end{equation}
Perform a change of variable $x=h\oplus y$ with $y\in \calL(\calK)$. 
Using Eq.~(\ref{proj7}) one can rewrite the above state as 
\begin{equation}
\label{proj9}
2^{-k/2-1}\sum_{y\in \calL(\calK)} e^{i\frac{\pi}4 q(h\oplus y)} \left(
1+ e^{i\frac{\pi}4 (\omega + \lambda(y))}\right)|h\oplus y\rangle,
\end{equation}
with
\begin{equation}
\label{proj10}
\omega=2m+4(\zeta, h)  + q(h\oplus \xi) - q(h)\in \{0,2,4,6\}
\end{equation}
and 
\begin{equation}
\label{proj11}
\lambda(y)=4(\zeta, y) + J(\xi,y)\in \{0,4\}.
\end{equation}
Let us first compute $\omega$. We have 
\begin{equation}
\label{xi}
\xi=\sum_{a=1}^k \xi_a g^a {\pmod 2}, \quad \xi_a=(\bar{g}^a,\xi).
\end{equation}
The decomposition Eq.~(\ref{xi}) can be computed in time 
 $O(kn)$. Once the coefficients $\xi_a$ are known, one can compute 
$\omega$ from 
\begin{equation}
\label{omega}
\omega= 2m + 4(\zeta,h) + \sum_{a=1}^k D_{a} \xi_a + \sum_{1\le a<b\le k} J_{a,b} \xi_a \xi_b.
\end{equation}
This takes time $O(kn)$.

Suppose first that $\omega\in \{0,4\}$. Then $e^{i\frac{\pi}4 \omega}=\pm 1$ and thus
\[
1+e^{i\frac{\pi}4 (\omega + \lambda(y))} = 
\left\{ \ba{rcl}
2 &\mbox{if} & \lambda(y)+\omega=0{\pmod 8} \\
0 &\mbox{if} & \lambda(y)+\omega=4{\pmod 8}. \\
\ea \right.
\]
We get 
\begin{equation}
\label{proj16a}
P_+|\calK,q\rangle=2^{-k/2}\sum_{x\in \calM}
e^{i\frac{\pi}4 q(x)}|x\rangle.
\end{equation}
where
\begin{equation}
\label{proj17a}
\calM=\calK\cap \{ x\in \FF_2^n \, : \, \lambda(h\oplus x)=\omega \}.
\end{equation}
Let us choose a vector $\gamma\in \FF_2^n$ such that 
$\lambda(y)=4(\gamma,y)$ for all $y\in \calL(\calK)$. We shall look for 
\begin{equation}
\label{eta0}
\gamma=\sum_{b=1}^k \eta_b \bar{g}^b {\pmod 2}, \quad \eta_b\in \{0,1\}.
\end{equation}
Choosing $y=g^a$ 
and using  $(g^a,\bar{g}^b)=\delta_{a,b}$ one gets 
\begin{equation}
\label{eta1}
4\eta_a=\lambda(g^a)=4(\zeta,g^a) + J(\xi,g^a), \quad 1\le a\le k.
\end{equation}
To compute $(\zeta,g^a)$ and $J(\xi,g^a)$ consider expansions
Eq.~(\ref{xi}) and 
\begin{equation}
\label{zeta}
\zeta=\sum_{a=1}^n \zeta_a \bar{g}^a {\pmod 2}, \quad \zeta_a=(g^a,\zeta).
\end{equation}
One can compute all the coefficients $\zeta_1,\ldots,\zeta_k$
in time $O(kn)$.  
The fact that $J(x,y)$ is a bilinear form implies 
\begin{equation}
\label{eta2}
4\eta_a=4\zeta_a +\sum_{b=1}^k J_{a,b} \xi_b, \quad 1\le a\le k.
\end{equation}
Thus $\eta_1,\ldots,\eta_k$ can be computed in time $O(kn)$.
Let $\omega=4\omega'$ with $\omega'\in \{0,1\}$.
We arrived at 
\[
\calM=\calK\cap \{x\in \FF_2^n \, : \,  (\gamma,x)=\alpha\}, \quad \alpha\equiv \omega' \oplus (\gamma,h).
\]
The standard form of
the state defined in Eq.~(\ref{proj16a}) can be computed 
by calling the function SHRINK$(|\calK,q\rangle, \gamma,\alpha)$,
see Appendix~B, which
takes time $O(kn)$.

Next suppose that $\omega\in \{2,6\}$. Then $e^{i\frac{\pi}4 \omega}=\pm i$ and thus
\[
1+e^{i\frac{\pi}4 (\omega + \lambda(y))} = 
\left\{ \ba{rcl}
\sqrt{2} e^{i\frac{\pi}4} &\mbox{if} & \lambda(y)+\omega=2{\pmod 8} \\
\sqrt{2} e^{-i\frac{\pi}4} &\mbox{if} & \lambda(y)+\omega=6{\pmod 8} \\
\ea \right.
\]
We shall choose a quadratic form $\lambda'\, : \, \calK\to \ZZ_8$ such that 
\begin{equation}
\label{lambda'}
\lambda'(h\oplus y)=\left\{ \ba{rcl} 0 &\mbox{if} & \lambda(y)=0,\\
2 &\mbox{if} & \lambda(y)=4.\\
\ea\right.
\end{equation}
Define 
\begin{equation}
\label{sigma}
\sigma =\left\{ \ba{rcl}
1 & \mbox{if} & \omega=2,\\
-1 & \mbox{if} & \omega=6.\\
\ea\right.
\end{equation}
Then the state in Eq.~(\ref{proj9}) can be written as
\begin{equation}
\label{proj12}
P_+|\calK,q\rangle=2^{-(k+1)/2}\sum_{x\in \calK} e^{i\frac{\pi}4 q'(x) } |x\rangle=2^{-1/2} |\calK,q'\rangle
\end{equation}
with a quadratic form 
\begin{equation}
\label{proj13}
q'(x)=\sigma + q(x) -\sigma \lambda'(x).
\end{equation}
To get the standard form of $|\calK,q'\rangle$ we need to
choose $\lambda'(x)$ satisfying Eq.~(\ref{lambda'}) and compute the coefficients
of $\lambda'(x)$ in the basis $g^1,\ldots,g^k$ of $\calL(\calK)$. 
First, let us compute the basis-dependent representation of $\lambda(y)$.
Suppose $y=\sum_{a=1}^k y_a g^a {\pmod 2}$
and let $\vec{y}=(y_1,\ldots,y_k)$.
Substituting Eqs.~(\ref{xi},\ref{zeta}) into Eq.~(\ref{proj11}) one  gets
\[
\lambda(\vec{y})=4\sum_{a=1}^k \eta_a y_a,
\]
where  $\eta_a\in \{0,1\}$ are defined by Eq.~(\ref{eta2}).
For any  $z_1,\ldots,z_k\in \{0,1\}$ one has the following identity:
\[
2(z_1\oplus \cdots \oplus z_k)=2\sum_{a=1}^k z_a -4 \sum_{1\le a<b\le k} z_a z_b {\pmod 8}.
\]
Choose $z_a=\eta_a y_a$ such that 
$\lambda(\vec{y})=4(z_1\oplus\cdots \oplus z_k)$.
Then a function $\lambda'(y)$ satisfying Eq.~(\ref{lambda'}) 
has a basis-dependent representation
$\lambda'(\vec{y})=2(z_1\oplus \cdots \oplus z_k)$, that is,
\begin{equation}
\label{proj14}
\lambda'(h\oplus y)=2\sum_{a=1}^k \eta_a y_a - 4 \sum_{1\le a<b\le k} \eta_a 
\eta_b \, y_b y_b.
\end{equation}
To summarize, the coefficients of the form $q'$ in the basis $g^1,\ldots,g^k$ are
$(Q',D',J')$, where
\begin{equation}
\label{proj15}
Q'=Q+\sigma, \quad D_a'=D_{a}-2\sigma \eta_a,
\end{equation}
and
\begin{equation}
\label{proj16}
J_{a,b}'=J_{a,b}+4\eta_a\eta_b \quad \mbox{for $a\ne b$}.
\end{equation}
This determines the standard form of $|\calK,q'\rangle$.

\noindent
{\em Case~2:} $\xi \notin \calL(\calK)$. 
Then $\xi\oplus x\notin \calK$ for any $x\in \calK$ and thus 
the states $|\calK,q\rangle$ and $P|\calK,q\rangle$
are supported on disjoint subsets of basis vectors. 
Define an affine space $\calM=\calL(\calM)\oplus h$
of dimension $k+1$, where $\calL(\calM)$ is spanned by $\calL(\calK)$
and $\xi$. We equip $\calL(\calM)$ with a basis
$g^1,\ldots,g^{k+1}$, where $g^{k+1}\equiv \xi$.
Then any vector $x\in \calM$ can be written in a basis-dependent
way as 
\[
x=h\oplus \sum_{a=1}^{k+1} x_a g^a {\pmod 2}.
\]
Let $\vec{x}=(x_1,\ldots,x_{k+1})$. 
A simple algebra shows that 
 \begin{equation}
\label{proj2}
P_+|\calK,q\rangle = 2^{-1-k/2} \sum_{x\in \calM} e^{i\frac{\pi}4 q'(x)} |x\rangle
=2^{-1/2} |\calM,q'\rangle,
\end{equation}
where $q'\, : \, \calM\to \ZZ_8$ is a quadratic form defined by
\begin{equation}
\label{qextend}
q'(\vec{x})= q(\vec{x}) + \left[ 2m + 4(\zeta,h\oplus \xi)\right] x_{k+1} +4\sum_{a=1}^k \zeta_a x_a x_{k+1}.
\end{equation}
Here it is understood that $q(\vec{x})$ depends only on the first $k$ coordinates
of $x$. Thus the coefficients of $q'$ in the chosen basis of $\calL(\calM)$ are
$Q'=Q$, $D'=[D,2m+4(\zeta,h\oplus \xi)]$, and  
\begin{equation}
\label{newJ}
J'=\left[ \ba{c|c}
J & 4\pmb{\zeta}^T \\
\hline
4\pmb{\zeta} & 4m  \\
\ea\right].
\end{equation}
Here $\pmb{\zeta}\equiv (\zeta_1,\ldots,\zeta_k)$ is a row vector.

It remains to compute the standard form of $\calM$.
Below we define  a function EXTEND$(\calK,\xi)$
that takes as input an affine space $\calK=(n,k,h,G,\bar{G})$ 
and a vector $\xi\in \FF_2^n$.
If $\xi \in \calL(\calK)$, the function does nothing.
Otherwise, the function  outputs
an affine space
$\calM=(n,k+1,h,H,\bar{H})$ such that
the first $k$ rows of $G$ and $H$ are the same
and the  $(k+1)$-th row of $H$
equals $\xi$. Since the function EXTEND
is very similar to the 
function SHRINK defined in Appendix~B,
we just state the algorithm skipping the analysis.

\begin{center}
\fbox{\parbox{0.9\linewidth}{
\begin{algorithmic}
\Function{$\calM$=Extend}{$\calK,\xi$}
\State{$S\gets \{ a\in [n]   \, : \,  (\xi,\bar{g}^a)=1\}$}
\State{$T\gets S\cap \{k+1,\ldots,n-1,n\}$}
\If{$T=\emptyset$}
\State{\Comment{$\xi\in \calL(\calK)$}}
\State{\Return{$\calK$}}
\EndIf
\State{Pick any $i\in T$}
\State{$S\gets S\setminus \{i\}$}
\For{$a\in S$}
\State{ $\bar{g}^a\gets \bar{g}^a \oplus \bar{g}^{i}$}
\EndFor
\State{$g^i \gets  g^i \oplus \sum_{a\in S} g^a$}
\State{\Comment{Now $g^i =\xi$}}
\State{Swap $g^i$ and $g^{k+1}$. Swap $\bar{g}^i$ and $\bar{g}^{k+1}$.}
\State{\Return{$(n,k+1,h,G,\bar{G})$}}
\EndFunction
\end{algorithmic}
}}
\end{center}
It has runtime $O(n^2)$.
We do not have to update the coefficients of $q'$ since
Eq.~(\ref{qextend}) defines $q'$ in the basis
$g^1,\ldots,g^k,\xi$ which coincides
with the new basis of $\calM$.
We conclude that the projected state $2^{-1/2}|\calM,q'\rangle$
can be computed in time $O(n^2)$.

Below we summarize the entire algorithm 
as a function MeasurePauli that takes as input a stabilizer state $|\calK,q\rangle\in \calS_n$ and a Pauli
operator $P\in \calP_n$. The function returns the norm of the projected
state $\Gamma=\| P_+|\calK,q\rangle \|$. If $\Gamma\ne 0$, the function
computes the standard form of the projected state $P_+|\calK,q\rangle$.
As before, we assume that the function can modify 
the data describing the input state.

\begin{center}
\fbox{\parbox{0.9\linewidth}{
\begin{algorithmic}
\Function{$\Gamma$=MeasurePauli}{$|\calK,q\rangle,P$}
\State{\Comment{$P=i^m Z(\zeta)X(\xi)$}}
\State{\Comment{$\calK=(n,k,h,G,\bar{G})$}}
\State{\Comment{$q=(Q,D,J)$}}
\For{$a=1$ to $k$}
\State{$\xi_a\gets (\bar{g}^a,\xi)$, $\zeta_a\gets (g^a,\zeta)$}
\EndFor
\State{$\xi'\gets \sum_{a=1}^k \xi_a g^a {\pmod 2}$}
\State{Compute $\omega\in \{0,2,4,6\}$ using Eq.~(\ref{omega})}
\If{$\xi'=\xi$ and $\omega\in \{0,4\}$}
\State{Compute $\eta_1,\ldots,\eta_k$  using Eq.~(\ref{eta2})}
\State{$\gamma\gets \sum_{a=1}^k \eta_a g^a {\pmod 2}$}
\State{$\omega'\gets \omega/4$}
\State{$\alpha \gets \omega'\oplus (\eta,h)$}
\State{$\epsilon \gets$SHRINK$(|\calK,q\rangle,\gamma,\alpha)$}
\If{$\epsilon=$EMPTY}
\State{$\Gamma \gets 0$}
\State{\Return}
\EndIf
\If{$\epsilon=$SAME}
\State{$\Gamma \gets 1$}
\State{\Return}
\EndIf
\If{$\epsilon=$SUCCESS}
\State{$\Gamma\gets 2^{-1/2}$}
\State{\Return}
\EndIf
\EndIf
\If{$\xi'=\xi$ and $\omega\in \{2,6\}$}
\State{$\sigma \gets 2-(\omega/2)$ }
\State{Compute $(Q',D',J')$ using Eqs.~(\ref{proj15},\ref{proj16})}
\State{$(Q,D,J)\gets (Q',D',J')$}
\State{$\Gamma \gets 2^{-1/2}$}
\State{\Return}
\EndIf
\If{$\xi'\ne \xi$}
\State{$\calK\gets$EXTEND$(\calK,\xi)$}
\State{$D\gets [D,2m+4(\zeta,h\oplus \xi)]$}
\State{$J\gets J'$, where $J'$ is defined in Eq.~(\ref{newJ})}
\State{$\Gamma\gets 2^{-1/2}$}
\State{\Return}
\EndIf
\EndFunction
\end{algorithmic}
}}
\end{center} 

\begin{table}[h]
\begin{tabular}{r|c|c|c|c|c|c}
Number of qubits & $\bf 10$ & $\bf 25$ & $\bf 50$ & $\bf 75$ & $\bf 100$ & $n$ \\
\hline
MeasurePauli & $0.27$ & $0.3$ & $0.4$ & $0.5$ & $0.6$ & $O(n^2)$ \\
\hline
RandomStabilizerState & $0.2$ & $0.3$ & $0.8$ & $1.7$ & $2.8$  & $O(n^2)$ \\
\hline
InnerProduct & $0.5$ & $1.5$ & $3.5$ & $6.5$ & $8.9$ &  $O(n^3)$ \\
\hline
ExponentialSum & $0.3$ &  $0.8$ & $2.2$ & $4.4$  &  $8$ & $O(n^3)$ \\
\end{tabular}
\caption{Average runtime in milliseconds for a MATLAB implementation of
our algorithms. Simulations were performed on a laptop with   2.6GHz Intel~i5 Dual Core CPU.
}
\label{table:1}
\end{table}

\section*{Appendix F: Simulation of the hidden shift algorithm}

\begin{figure*}
\includegraphics[height=4cm]{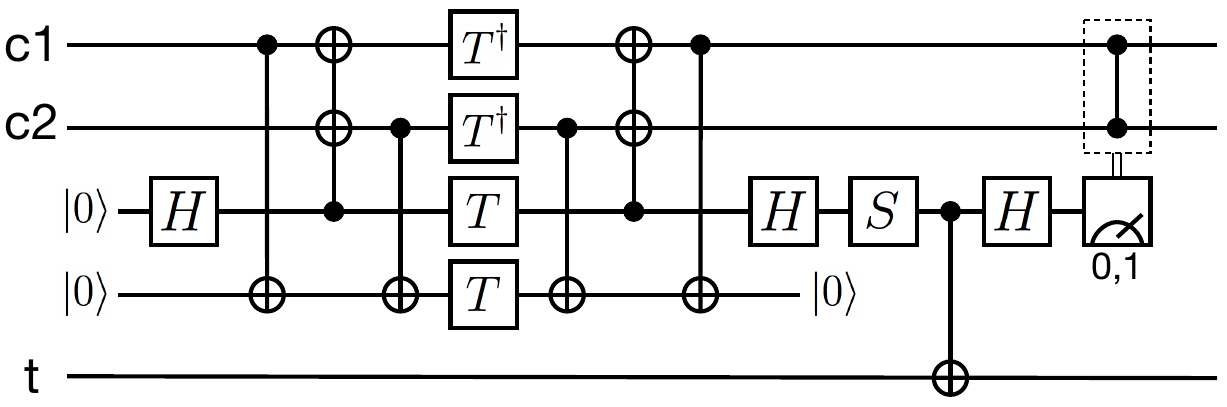}
\caption{Gadget from Ref.~\cite{Jones2013} implementing the Toffoli gate.
The two control qubits and the target qubit are denoted $c_1,c_2$ and $t$ respectively. Both measurement outcomes appear with probability $1/2$.
The final controlled-$Z$ gate on qubits $c_1,c_2$ is applied only if the measurement outcome is `$1$'.}
\label{fig:Toffoli}
\end{figure*}

Here we provide further details of the simulations reported in Fig.~\ref{fig:plots}. 
Recall that we  simulate a circuit
\begin{equation}
\label{HSP}
U= H^{\otimes n} O_{f'} H^{\otimes n} O_{f}H^{\otimes n},
\end{equation}
where $O_f |x\rangle =f(x)|x\rangle$ and
$O_{f'}|x\rangle=f'(x)|x\rangle$ are oracle circuits
for some bent functions 
$f,f'\,:\, \FF_2^n \to \{+1,-1\}$ such that 
\begin{equation}
f'(x)=2^{-n/2} \sum_{y\in \FF_2^n} (-1)^{x\cdot y}  f(y\oplus s) \quad \mbox{for all $x\in \FF_2^n$}.
\label{hadamard1}
\end{equation}
Here  $s\in \FF_2^n$ is the hidden shift that can be found from $|s\rangle=U|0^n\rangle$.
In our simulations the hidden shift $s$
was chosen at random from the uniform distribution. 
The  function $f$ was chosen from (a subclass of) the Maiorana McFarland family of bent functions. 
In general, a Maiorana McFarland bent function is defined as follows. Suppose $n$ is even. Let 
\[
g:\mathbb{F}_2^{n/2}\rightarrow \mathbb{F}_2 \quad \mbox{and} \quad
\pi:\mathbb{F}_2^{n/2}\rightarrow \mathbb{F}_2^{n/2}
\]
be any Boolean function and any permutation respectively.
 For any such pair $g,\pi$ we may define a bent function $f:\mathbb{F}_2^{n}\rightarrow \{+1,-1\}$
 according to
\begin{equation}
f(x,y)=(-1)^{g(x)+y\cdot \pi(x)} \qquad x,y\in \mathbb{F}_2^{n/2}.
\label{eq:fMM}
\end{equation}
The Hadamard transform of $f$ is given by
\begin{equation}
\label{f-dual}
 2^{-n/2} \sum_{u,v} (-1)^{u\cdot x + v\cdot y} f(u,v)
 =(-1)^{x\cdot \pi^{-1}(y)+g(\pi^{-1}(y))}.
\end{equation}
In our simulations we only used bent functions of the form Eq.~(\ref{eq:fMM}) with  $\pi=I$ (the identity permutation).
The Boolean function $g$ was chosen at random, as explained below.
Letting $O_g$ be the $n/2$-qubit diagonal unitary
\[
O_g|x\rangle=(-1)^{g(x)}|x\rangle \quad x\in \mathbb{F}_2^{n/2}
\]
we see that a quantum circuit which implements the $n$-qubit unitary oracle $O_f|x,y\rangle=f(x,y)|x,y\rangle$ can be decomposed as
\[
O_f=\left(\prod_{i=1}^{n/2} CZ_{i,i+n/2}\right)O_g\otimes I
\]
where $CZ=\mathrm{diag}(1,1,1,-1)$ is the two-qubit controlled-$Z$ gate. Here the tensor product separates the first $n/2$ qubits from the last $n/2$. Likewise, from Eqs.~(\ref{hadamard1},\ref{f-dual}) one infers that 
\[
O_{f'}=\left[  \left(\prod_{i=1}^{n/2} CZ_{i,i+n/2}\right) I\otimes O_g \right] Z(s)
\]
Note that the total $T$-count of the circuit $U$ is twice the $T$-count of $O_g$.
To construct a circuit implementing $O_g$  we chose a sequence of gates from the set
$\{Z,CZ,CCZ\}$, where  $CCZ$ is the controlled-controlled-Z gate.
We first fixed the number of $CCZ$ gates (five and six for the simulations reported in the left/right plots of Fig.~\ref{fig:plots} respectively), and then produced a  circuit $O_g$ alternating the $CCZ$ gates (on a randomly chosen triple of qubits) with random sequences of $200$ Clifford gates from the set $\{Z,CZ\}$.  
Note that the $CCZ$ gate  can be replaced by the Toffoli gate using the 
identity
\begin{equation}
 CCZ=(I\otimes I\otimes H) \mathrm{Toff} (I\otimes I\otimes H).
\label{eq:CCZ}
\end{equation}
To decompose Toffoli gates into Clifford and $T$-gates
we used a gadget proposed by Jones~\cite{Jones2013}, see Fig.~\ref{fig:Toffoli}.
The gadget uses four $T$-gates, two ancillary qubits initialized in the state $|0\rangle$,
several Clifford gates, and the $0,1$-measurement. The final Clifford gate is classically controlled by 
the measurement outcome. 
To simulate the gadget we use the trick described in the remark between Eqs.~(\ref{conditional}, \ref{AproxT}).
Namely, in our simulation the measurement of the ancillary qubit is replaced by postselection on a random output bit $y$ in exactly the same way as was done for the $T$-gate gadget. The second ancilla in the gadget is never measured and is returned to the state $|0\rangle$ at the output; this ancilla is reused by all Toffolis in the circuit.

The simulation algorithm we implemented differs in some small details from the algorithm analyzed in the main text of the paper. To produce each data point in Fig.~\ref{fig:plots} we first fixed the output qubit $q\in \{1,2,\ldots,40\}$. We then estimated the ratio (cf. Eq.~(\ref{tildePa}))
\begin{equation}
P^{y}_{out}(1) = \frac{ \langle 0^{N}\otimes \psi| V_y^\dag (|1\rangle\langle 1|_q \otimes |y\rangle\langle y|) V_y |0^{N} \otimes \psi\rangle}
{ \langle 0^{N}\otimes  \psi| V_y^\dag (I_{n+1} \otimes |y\rangle\langle y|) V_y |0^{N}\otimes \psi\rangle }.
\label{eq:Pout}
\end{equation}
for a randomly chosen postselection bit-string $y$. For us $N=n+n_{anc}$ where $n=40$ is the number of qubits in the original circuit to be simulated while $n_{anc}$ is the number of ancillae initialized in the state $|0\rangle$ which are used for the Toffoli gadgets. Since each Toffoli gadget requires two ancilla, one of which is shared by all of them, we have $n_{anc}=1+\mathrm{TF}$ where $\mathrm{TF}$ is the number of Toffoli gadgets used. In Eq.~(\ref{eq:Pout}) the number of postselection bits is $|y|=t+\mathrm{TF}$ where $t$ is the number of $T$-gates in the circuit (including the four $T$-gates within each Toffoli gadget). The unitary $V_y$ is a $(n+1+\mathrm{TF}+t)$-qubit Clifford unitary which is obtained by replacing all Toffoli gadgets and $T$ gate gadgets by the appropriate Clifford circuits obtained by postselecting on the measurement outcomes defined by the bit string $y$. Finally, the state $\psi$ in Eq.~(\ref{eq:Pout}) is a $t$-qubit state which approximates $t$ copies of the magic state $|A\rangle^{\otimes t}$. In particular, $\psi$ was derived from a $k$-dimensional subspace $\cal{L}$ of $\mathbb{F}_2^{t}$ in the manner described in the main text of the paper. In our simulations we used $k=11$ (left plot in Fig.~\ref{fig:plots}) and $k=12$ (right plot in Fig.~\ref{fig:plots}). The fidelities were $|\langle A^{\otimes t}|\psi\rangle|\approx 0.81$ and $|\langle A^{\otimes t}|\psi\rangle|\approx 0.69$ respectively.

To estimate $P^{y}_{out}(1)$ we computed integers $u,v$ and stabilizer groups $\cal{F},\cal{G}$ such that
\begin{equation}
\langle 0^{N}\otimes \psi| V_y^\dag (|1\rangle\langle 1|_q \otimes |y\rangle\langle y|) V_y |0^{N} \otimes \psi\rangle=2^{-u}\langle \psi |\Pi_{\cal{F}}|\psi\rangle
\label{eq:oneprob}
\end{equation}
and 
\begin{equation}
\langle 0^{N}\otimes \psi| V_y^\dag (|0\rangle\langle 0|_q \otimes |y\rangle\langle y|) V_y |0^{N} \otimes \psi\rangle=2^{-v}\langle \psi |\Pi_{\cal{G}}|\psi\rangle
\label{eq:zeroprob}
\end{equation}
and then, if $\Pi_{\cal{F}}\neq 0$ and $\Pi_{\cal{G}}\neq 0$, we computed approximations $\alpha,\beta$ to the quantities Eqs.~(\ref{eq:zeroprob},\ref{eq:oneprob}) using the norm estimation procedure described in the main text. The number of random stabilizer states sampled by the norm estimation procedure was chosen to be $100$ (left plot in Fig.~\ref{fig:plots}) or $50$ (right plot in Fig.~\ref{fig:plots}). Our estimate of $P_{out}^y(1)$ was then $\alpha/(\alpha+\beta)$ (cf. Eq.~(\ref{eq:Pout})). Note that if either $\Pi_{\cal{F}}=0$ or $\Pi_{\cal{G}}=0$ then $\alpha,\beta$, and $P_{out}^y(1)$ can be computed without ever calling the norm estimation subroutine. This special case occured for all qubits $1,2,\ldots, 20$ in both our simulations (as well as for some of the other data points).

\bibliographystyle{apsrev}
\bibliography{mybib}

\end{document}